\documentclass[11pt,a4paper]{article}

\usepackage[ruled,linesnumbered,vlined,boxed,commentsnumbered]{algorithm2e}

\usepackage{amssymb}
\usepackage{comment}
\usepackage{amsmath}
\usepackage{authblk}
\usepackage{amsthm}
\usepackage{color}
\usepackage{a4wide}
\usepackage{graphicx}
\usepackage{mathabx}
\usepackage{microtype}
\def\lf{\tiny}

\def\nnll{\refstepcounter{linenumber}\lf\thelinenumber}
\newcounter{linenumber}

\newcommand{\commentline}[1]{\hspace{1cm}\{ \textit{#1} \}}

\usepackage{url}
\usepackage{cite}
\newtheorem{theorem}{Theorem}[section]
\newtheorem{lemma}{Lemma}[section]
\newtheorem{property}{Property}[section]

\definecolor{heraldBlue}{rgb}{0.0,0.0,0.8}
\definecolor{heraldRed}{rgb}{0.8,0.0,0.0}
\definecolor{heraldGray}{rgb}{0.4,0.4,0.4}
\definecolor{heraldBlack}{rgb}{0.0,0.0,0.0} 
\definecolor{heraldGreen}{rgb}{0.0,0.4,0.0} 

\def\NOTES{1}
\def\SAVESPACE{0}

\if \SAVESPACE 1
\usepackage{titlesec}
\titlespacing\section{0pt}{7pt}{6pt}

\fi

\if \NOTES 1
\newcommand{\mcnote}[1]{\textcolor{heraldBlue}{\small \bf [MC: #1]}}
\newcommand{\dlnote}[1]{\textcolor{heraldBlue}{\small \bf [DL: #1]}}
\newcommand{\ssnote}[1]{\textcolor{heraldBlue}{\small \bf [SS: #1]}}
\newcommand{\pknote}[1]{\textcolor{heraldBlue}{\small \bf [PK: #1]}}
\newcommand{\atnote}[1]{\textcolor{heraldBlue}{\small \bf [AT: #1]}}
\else
\newcommand{\mcnote}[1]{}
\newcommand{\dlnote}[1]{}
\newcommand{\ssnote}[1]{}
\newcommand{\pknote}[1]{}
\newcommand{\atnote}[1]{}
\newcommand{\problem}[1]{}
\fi

\newcommand{\true}{\textit{true}}
\newcommand{\false}{\textit{false}}

\newcommand{\ignore}[1]{}

\newcommand{\cO}{\mathcal O}

\newcommand{\cC}{\mathcal C}
\newcommand{\cS}{\mathcal S}

\newcommand{\Lat}{\mathcal L}

\newcommand{\myparagraph}[1]{\vspace{3.5pt}\noindent \textbf{#1}}

\title{Reconfigurable Lattice Agreement and Applications}
\author[1]{Petr Kuznetsov}
\author[2]{Thibault Rieutord}
\author[2]{Sara Tucci-Piergiovanni}
\affil[1]{LTCI, T\'el\'ecom Paris, Institut Polytechnique Paris}
\affil[2]{CEA LIST, PC 174, Gif-sur-Yvette, 91191, France}
\date{}                     

\begin{document}






\maketitle




\begin{abstract}
  Reconfiguration is one of the central mechanisms in distributed systems.
  Due to failures and connectivity disruptions, the very set of
  service replicas (or \emph{servers}) and their roles in the
  computation may have to be reconfigured over time. To provide the desired level of
  consistency and availability to applications running on top of these
  servers, the \emph{clients} of the service should be able to reach some form of
  agreement on the system configuration.
  We observe that this agreement is naturally captured via a \emph{lattice} partial order on the system
  states.  
  We propose an asynchronous implementation of \emph{reconfigurable}
  lattice agreement that implies elegant reconfigurable versions of a
  large class of \emph{lattice} abstract data types, such as
  max-registers and conflict detectors,  as well as
  popular distributed programming abstractions, such as atomic 
  snapshot and commit-adopt. 

\vspace{1em}  
  
  \textbf{\underline{keywords:}} Reconfigurable services, lattice agreement.
\end{abstract}

\section{Introduction}
\label{sec:intro}

A decentralized service~\cite{Sch90,ABD95,Lam98,GGKV09} runs on a set of fault-prone
\emph{servers} that store replicas of the system state and run
a synchronization protocol to ensure consistency of concurrent data accesses.
In the context of a storage system exporting read and write operations, 
several proposals~\cite{rambo,dynastore,parsimonious,smartmerge,freestore,SKM17-reconf}
came out with a reconfiguration interface that allows the servers to join and leave
while ensuring consistency of the stored data.
Early proposals of reconfigurable storage systems~\cite{rambo} were based on using \emph{consensus}~\cite{FLP85,Her91} 
to ensure that replicas \emph{agree} on the evolution of the system membership.
Consensus, however, is expensive and difficult to implement, and recent 
solutions~\cite{dynastore,parsimonious,smartmerge,freestore,SKM17-reconf} replace consensus with
weaker abstractions capturing the minimal coordination required to
safely modify the system configuration.
These solutions, however,  lack a uniform way of deriving reconfigurable versions of static objects.   

\myparagraph{Lattice objects.}
In this paper, we propose a universal reconfigurable construction for a 
large class of objects.
Unlike a consensus-based reconfiguration proposed earlier for
generic state-machine replication~\cite{paxos-reconfigure}, our
construction is purely asynchronous, at the expense of assuming a restricted
object behavior.
More precisely, we assume that  the set $\Lat$ of the object's states
can be represented as a (join semi-) \emph{lattice}  $(\Lat, \sqsubseteq)$, where $\Lat$ 
is partially ordered by the binary relation $\sqsubseteq$ such that for
all elements of $x,y\in \Lat$, there exists the \emph{least upper bound} in~$\Lat$,
denoted $x\sqcup y$, where $\sqcup$, called the \emph{join} operator, is an associative, commutative, and idempotent binary
operator on $\Lat$. 
Many important data types, such as sets and counters,
as well as useful concurrent abstractions, such as conflict detector~\cite{AE14}, can be  expressed this way.  
Intuitively,~$x\sqcup y$ can be seen as a \emph{merge} of two
alternatively proposed updated states~$x$ and~$y$.
Thus, an implementation ensuring that all
``observable''  states are ordered by
$\sqsubseteq$ cannot be distinguished from an atomic object.   

Consider, for example, the \emph{max-register}~\cite{max-register}
data type with two operations: \textit{writeMax} writes a
value and \textit{readMax} returns the largest value written so far.
Its state space can be represented as a lattice $(\sqsubseteq,\sqcup)$
of its values, where $\sqsubseteq=\leq$ and $x \sqcup y = \max(x,y)$.
Intuitively, a linearizable implementation of max-register must ensure
that every read value is a join of previously proposed values, and all
read values are totally ordered (with respect to $\leq$).     
\ignore{
\footnote{Here we assume
that updates on a given position in the atomic-snapshot object are
totally ordered, which is the case when each position has a distinct
dedicated writer. Our approach can also be extended 
to the multi-writer case.}
Concurrent updates on distinct positions in an atomic can then be merged in a state in
which both updates took place using the simple set-union operation.
Therefore, if we ensure that all observable states of the object are related by
containment, the object is indistinguishable from atomic snapshot.     
}

\myparagraph{Reconfigurable lattice agreement.}
%
In this paper, we introduce the \emph{reconfigurable lattice agreement}~\cite{lattice-hagit,gla}. 
It is natural to treat the \emph{system configuration}, i.e., the set of servers available
for data replication, as an element in a lattice.   
A lattice-defined join of configurations, possibly concurrently proposed by different
clients, results in a new configuration.
The lattice-agreement protocol ensures that configurations evaluated
by concurrent processes are \emph{ordered}.
Despite processes possibly disagreeing about the precise configuration
they belong to, they can use the configurations relative ordering to 
maintain the system data consistency.

A configuration is defined by a set of servers, a quorum system~\cite{quorums}, 
i.e., a set system ensuring the intersection
property\footnote{The most commonly used quorum system is 
  majority-based: quorums are all majorities of servers. 
  We can, however, use any other quorum system, as suggested in~\cite{rambo, smartmerge}.}
and, possibly, other parameters.
For example,  elements of a reconfiguration lattice can be defined as sets of
\textit{configuration updates}: each such update either adds a server to the
configuration or removes a server from it. 
The \emph{members} of such a configuration are the set of all
servers that were added but not yet removed. 
A join of two configurations defined this way is simply a union of their
updates  (this approach is implicitly used in earlier asynchronous
reconfigurable constructions~\cite{dynastore,parsimonious,SKM17-reconf}).

\myparagraph{Reconfigurable L-ADTs and applications.}
We show that our reconfigurable lattice agreement, defined on a
product of a \emph{configuration lattice} and an \emph{object
  lattice},   immediately implies reconfigurable versions of 
many sequential types, such as \emph{max-register} and
\emph{conflict detector}.
More generally, any \emph{state-based commutative} abstract data
(called \emph{L-ADT}, for \emph{lattice abstract data type}, in this paper)  has a reconfigurable
\emph{interval-linearizable}~\cite{interval-linearizability}
implementation. Intuitively, interval-linearizability~\cite{interval-linearizability},
a generalization of  the classical linearizability~\cite{HW90},
allows specifying the behavior of an object when
multiple concurrent operations ``influence'' each other.
Their effects are then merged using a join operator, which turns out to be natural in the
context of reconfigurable objects.

Our transformations are straightforward.
To get an (interval-linearizable) reconfigurable implementation of an
L-ADT, we simply use its state lattice, as a parameter, in our reconfigurable lattice agreement.
%
%
%
%
The resulting implementations are naturally composable: we get a reconfigurable
composition of two L-ADTs by using a \emph{product} of their lattices.  
If operations on the object can be partitioned
into \emph{updates} (modifying the object state without providing informative
responses) and \emph{queries} (not modifying the object state), as in
the case of max-registers, the
reconfigurable implementation becomes linearizable\footnote{Such 
   ``update-query'' L-ADTs are known as state-based convergent
  replicated data types (CvRDT)~\cite{crdt}. These include
  \textit{max-register}, \textit{set} and \textit{abort flag} (a new
  type introduced in this paper).}.

We then use our reconfigurable implementations of \textit{max-register}, \emph{conflict
detector}, \emph{set} and \emph{abort-flag} to devise  
reconfigurable versions of \emph{atomic snapshot}~\cite{AADGMS93},
\emph{commit-adopt}~\cite{Gaf98} and \emph{safe agreement}~\cite{BG93b}. 
Figure~\ref{fig:summary} shows how are constructions are related.

\begin{figure}[tbph]
  \centering
 \includegraphics[scale=0.85]{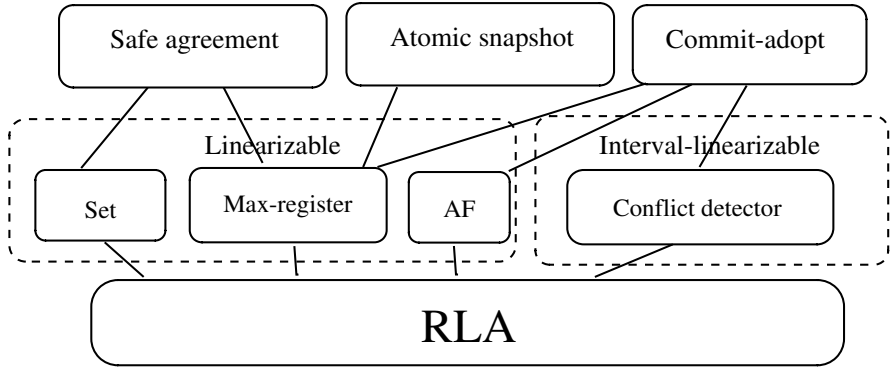}
  \caption{Our reconfigurable implementations: reconfigurable lattice
    agreement (RLA) is used to construct
    linearizable implementations of a set, a max-register, an abort
    flag, and an interval-linearizable implementation of a conflict detector. On
  top of max-registers we construct an atomic snapshot; on top of
  a max-register, an abort-flag, and a conflict detector, we construct a commit-adopt
  abstraction; and, on top of sets and a max register, we implement a safe agreement abstraction.}
  \label{fig:summary}
\end{figure}

\myparagraph{Summary.}
Our reconfigurable construction is the first to be, at the same time:
\begin{itemize}
\item Asynchronous, unlike consensus-based
solutions~\cite{paxos-reconfigure,CL02,rambo}, and not assuming an
external lattice agreement service~\cite{smartmerge};
        
\item Uniformly applicable to a large class of objects, unlike existing
reconfigurable systems that either focus on read-write
storage~\cite{rambo,dynastore,parsimonious,smartmerge} or require data type-specific implementations of exported reconfiguration interfaces~\cite{SKM17-reconf};

\item Allowing for a straightforward composition of reconfigurable
  objects; 

\item Maintaining configurations with
abstract \emph{quorum systems}~\cite{quorums}, not restricted to
\emph{majority-based} quorums~\cite{dynastore,parsimonious};

\item Exhibiting optimal time complexity and 
message complexity comparable with the best known
implementations~\cite{dynastore,smartmerge,SKM17-reconf};

\item Logically separating \emph{clients} (external entities that use the
implemented service) from \emph{servers} (entities that maintain the
service and can be reconfigured).
\end{itemize}

We also believe our reconfigurable construction to be the simplest on
the market, using only twenty one lines of
pseudocode and provided with a concise proof.   

\myparagraph{Roadmap.}
The rest of the paper is organized as follows.
We give basic model definitions in Section~\ref{sec:model}. 
In Section~\ref{sec:ladtDef}, we define our type of 
reconfigurable objects, followed by the related notion of
reconfigurable lattice agreement in Section~\ref{sec:definition}.
In Section~\ref{sec:implementation}, we describe our implementation
of reconfigurable lattice agreement, and, in Section~\ref{sec:ladt}, 
we show how to use it to implement a reconfigurable L-ADT object.   
In Section~\ref{sec:applications} we describe some possible applications.
We conclude with, in Section~\ref{sec:related}, an overview of the related 
work, and, in Section~\ref{sec:conc}, a discussion on algorithms complexity 
and possible trade-offs.

\section{Definitions}
\label{sec:model}

\myparagraph{Replicas and clients.}
Let $\Pi$ be a (possibly infinite) set of potentially participating processes.
A subset of the processes, called \emph{replicas}, are used to maintain a replicated object.
%
%
A process can also act as a \emph{client}, invoking operations on the
object and proposing system reconfigurations. 
Both replicas and clients are subject to crash failures: a process
\emph{fails} when it prematurely stops taking steps of its algorithm.
A \emph{failure model} stipulates when and where failures might occur.
We present our failure model in Section~\ref{sec:definition}, where we formally
define reconfigurable lattice agreement.

%


\myparagraph{Abstract data types.}
An abstract data type ($\textit{ADT}$) is a tuple
$T=(A,B,Z,z_0,\tau,\delta)$.
Here~$A$ and $B$ are countable sets called the \emph{inputs} and
\emph{outputs}.  $Z$ is a countable set of abstract object 
\emph{states}, $z_0\in Z$ being the initial state of the object. 
The map~$\tau: Z\times A \to Z$ is the \emph{transition function}, specifying the effect of an input on the object state 
and the map~$\delta: Z\times A \to B$ is the \emph{output function}, specifying the output returned for a given input and 
object local state. The input represents an operation with its parameters, where (i) the operation can have a side-effect that changes the abstract state according to transition function $\tau$  and (ii) the operation can return values taken in the output $B$, which depends on the state in which it is called and the output function $\delta$ (for simplicity, we only consider deterministic types here, check, e.g.,~\cite{perrinBook}, for more details.)

\myparagraph{Interval linearizability.}
We now briefly recall the notion of
\emph{interval-linearizability}~\cite{interval-linearizability}, a
recent generalization of linearizability~\cite{HW90}.    

Let us consider an abstract data type $T=(A,B,Z,z_0,\tau,\delta)$.
A \emph{history} of $T$ is a sequence of inputs (elements of $A$) and outputs (elements
of $B$), each labeled with a process identifier and an
operation identifier.
An \emph{interval-sequential history} 
is a sequence:
\[
z_0,I_1,R_1,z_1,I_2,R_2,z_2\ldots,I_m,R_m,z_m,
\]  
where each $z_i\in Z$ is a state, $I_i\subseteq A$ is a set of inputs,
and $R_i\subseteq B$ is a set of outputs.
An \emph{interval-sequential specification}  is a set of
interval-sequential histories.

We only consider \emph{well-formed} histories. Informally, in a
well-formed history, a process only invokes an operation once its previous operation has
returned and every response $r$ is preceded by a ``matching'' operation $i$.   

\ignore{
A projection of a history or an interval sequential history $H$ to
inputs and outputs labelled with $p$ is denoted by $H|p$.
We only consider \emph{well-formed} histories, i.e., informally, a
process only invokes an operation once

An input is \emph{complete} in $H$ if it is followed with a \emph{matching
response}, i.e.,      
}

A history $H$ is \emph{interval-linearizable} respectively to an
interval-sequential specification~$\cS$ if it can be \emph{completed} (by
adding matching responses to incomplete operations) so that the
resulting history $\bar H$ can be associated with an
interval-sequential history $S$ such that: (1)~$\bar H$ and $S$ are
\emph{equivalent}, i.e., $\forall p\in\Pi$, $\bar H|p=S|p$, (2)~$S\in
\cS$, and (3)~$\rightarrow_{H}\subseteq \rightarrow_{S}$, i.e., $S$
preserves the real-time precedence relation of $H$. 
(Check~\cite{interval-linearizability} for more details on the definition.)

\myparagraph{Lattice agreement.}
An abstract (join semi-)\emph{lattice} is a tuple  $(\Lat, \sqsubseteq)$, where $\Lat$ is a set
partially ordered by the binary relation $\sqsubseteq$ such that for
all elements of $x,y\in \Lat$, there exists the least upper bound for
the set $\{x,y\}$.
The least upper bound is an associative, commutative, and idempotent binary
operation on $\Lat$, denoted by $\sqcup$ and called the \emph{join
  operator} on $\Lat$.
We write $x\sqsubset y$ whenever $x\sqsubseteq  y$ and $x\neq y$.
With a slight abuse of notation, for a set $L\subseteq \Lat$, we also write $\bigsqcup L$ for $\bigsqcup_{x\in
L} x$, i.e., $\bigsqcup L$ is the join of the elements of $L$. 

Notice that two lattices $(\Lat_1, \sqsubseteq_1)$ and $(\Lat_2,
\sqsubseteq_2)$ naturally imply a \emph{product} lattice
$(\Lat_1\times\Lat_2, \sqsubseteq_1\times\sqsubseteq_2)$ with a product
join operator $\sqcup=\sqcup_1\times\sqcup_2$.  
Here for all $(x_1,x_2),(y_1,y_2)\in\Lat_1\times\Lat_2$,
$(x_1,x_2)(\sqsubseteq_1\times\sqsubseteq_2)(y_1,y_2)$ if and only if
$x_1\sqsubseteq_1y_1$ and $x_2\sqsubseteq_2y_2$.

The (generalized) \emph{lattice agreement} concurrent abstraction, defined on a lattice
$(\Lat, \sqsubseteq)$, exports a single operation \textit{propose}
that takes an element of $\Lat$ as an argument and returns an element
of $\Lat$ as a response.
When the operation $\textit{propose}(x)$ is invoked by process $p$ we
say that $p$ \emph{proposes} $v$, and when the operation returns $v'$
we say that $p$ \emph{learns} $v'$.
Assuming that no process invokes a new operation before its previous
operation returns, the abstraction satisfies the following properties:
\begin{itemize}
\item {\bf Validity.} If a $\textit{propose}(v)$ operation returns a value $v'$ then 
  $v'$ is a join of some proposed values including $v$ and all values
  learnt before the invocation of the operation.
  

\item {\bf Consistency.} The learnt values are totally ordered by
  $\sqsubseteq$.

\item {\bf Liveness.} If a process invokes a \textit{propose}
  operation and does not fail then the operation eventually returns.
  
\end {itemize}  
\textbf{A historical remark.} The  original definition of long-lived lattice agreement~\cite{gla}
    separates ``receive'' events and ``learn''
    events.
Here we suggest a simpler definition that represents the two events as the invocation and 
  the response of a \textit{propose} operation.
This also allows us to slightly strengthen the validity condition so
that it accounts for the \emph{precedence} relation between
\textit{propose} operations.   
As a result, we can directly relate lattice agreement to linearizable~\cite{HW90} and
interval-linearizable~\cite{interval-linearizability} implementations,
without introducing artificial ``nop'' operations~\cite{gla}.

\section{Lattice Abstract Data Type}
\label{sec:ladtDef}

%
In this section, we introduce a class of types that we call 
\emph{lattice abstract data types} or $\textit{L-ADT}$.
In an $\textit{L-ADT}$, the set of states forms a join semi-lattice 
with a partial order $\sqsubseteq^Z$. 
A lattice object is therefore defined as a tuple 
$L=(A,B,(Z,\sqsubseteq^Z,\sqcup^Z),z_0,\tau,\delta)$.\footnote{For
convenience, we explicitly specify the join operator $\sqcup^Z$ here, i.e., the least upper
bound of $\sqsubseteq^Z$.}  
Moreover, the transition function $\delta$ must comply 
with the partial order $\sqsubseteq^Z$, that is $\forall z,a\in Z\times A: z\sqsubseteq^Z \tau(z,a)$, and 
the composition of transitions must comply with the join operator, that is 
$\forall z\in Z,\forall a,a'\in A: \tau(\tau(z,a),a') =\tau(z,a)\sqcup^Z\tau(z,a') =\tau(\tau(z,a'),a)$.
Hence, we can say that the transition function is ``commutative''. 

\myparagraph{Update-query L-ADT.}
We say an L-ADT $L=(A,B,(Z,\sqsubseteq^Z,\sqcup^Z),z_0,\tau,\delta)$
is \emph{update-query}  if $A$ can be partitioned in \emph{updates} $U$ and \emph{queries} $Q$
such that:
\begin{itemize}
\item there exists a special ``dummy'' response $\bot$ ($z_0$ may also be used) 
such that $\forall u\in U,\; z\in Z$, $\delta(u,z)=\bot$, i.e., updates do not
return informative responses;  
\item $\forall q\in Q,\; z\in Z$, $\tau(q,z)=z$, i.e., queries do
not modify the states.
\end{itemize}
This class is also known as a state-based convergent replicated
 data types (CvRDT)~\cite{crdt}.
Typical examples of update-query L-ADTs are
\textit{max-register}~\cite{max-register} (see
Section~\ref{sec:intro}) or \emph{sets}.
Note that any (L-)ADT can be transformed into an update-query (L-)ADT by 
``splitting its operations'' into an update and a query (see~\cite{perrinBook}).

\myparagraph{Composition of L-ADTs.}
The composition of two ADTs $T=(A,B,Z,z_0,\tau,\delta)$ and $T'=(A',B',Z',z_0',\tau',\delta')$ 
is denoted $T\times T'$ and is equal to $(A+A', B\cup B', Z\times Z',(z_0,z_0'),\tau'',\delta'')$;
where $A+A'$ denotes the disjoint union and where $\tau''$ and $\delta''$ apply, according to the domain~$A$ or $A'$ of the input, 
either $\tau$ and $\delta$ or $\tau'$ and $\delta'$ on their respecting half of the state (see~\cite{perrinBook}).

Since the cartesian product of two lattices remains a lattice, the composition of L-ADTs is naturally defined and 
produces an L-ADT. The composition is also closed to update-query~ADT, and thus to update-query L-ADT. 
Moreover, the composition is an associative and commutative operator, and hence, can easily be used to construct elaborate L-ADT.

\ignore{
Given two lattice, their product is also a lattice. Hence, as the composition of two objects provides an abstract 
state alphabet that is the product of each alphabet, then the composition of two $\textit{L-ADT}$ creates an 
$\textit{L-ADT}$.

Note tha we can also merge operations, that is to define an operation on the product as a product of operation on each 
objects. This can be more convenient. 

TO BE detailled with structure and using references from ADT and Lattice.
}

\myparagraph{Configurations as L-ADTs.}
%
Let us also use the formlism of L-ADT to define a \emph{configuration L-ADT} as a tuple $(A^{\cC},B^{\cC},(\cC,\sqsubseteq^{\cC},\sqcup^{\cC}),C_0,\tau^{\cC},\delta^{\cC})$ 
with $C_0\in \cC$ the \emph{initial configuration}.  
%
%
For each element  $C$ of the \emph{configuration lattice} $\cC$, the input set $A$ includes the query operations 
$\textit{members}()$, such that $\delta^{\cC}(C,\textit{members}()) \subseteq \Pi$, and $\textit{quorums}()$ 
where $\delta^{\cC}(C,\textit{quorums}())\subseteq 2^{\delta^{\cC}( C,\textit{members}())}$ is a \emph{quorum system}, 
that is, every two subsets in~$\delta^{\cC}(C,\textit{quorums}())$ have a non-empty intersection.  
With a slight abuse of notation, we will write these operations as $\textit{members}(C)$ and $\textit{quorums}(C)$.


For example,  $\cC$ can be the set of  tuples $(In,Out)$, where $In\subseteq \Pi$ is a set
of \emph{activated} processes,  and $Out\subseteq \Pi$ is a set of \emph{removed} processes. 
Then $\sqsubseteq^\cC$ can be defined as the piecewise set inclusion on $(In,Out)$. 
The set of members of $(In,Out)$ will simply be $In-Out$ and 
the set of quorums (pairwise-intersecting subsets of $In-Out$), e.g., all majorities of~$In-Out$. 
Operations in $A^{\cC}$ can be $\textit{add}(s)$, $s\in\Pi$, that adds $s$ to the
set of activated processes and
$\textit{remove}(s)$, $s\in\Pi$, that adds $s$ to the set of removed
processes of a configuration.    
One can easily see that updates ``commute'' and that the type is indeed
a configuration L-ADT. Let us note that L-ADTs allow for more expressive
reconfiguration operations than simple \textit{adds} and
\textit{removes}, e.g., maintaining a minimal number of members in a
configuration or adapting the quorum system dynamically, as 
studied in detail by Jehl et al. in~\cite{smartmerge}.

\myparagraph{Interval-sequential specifications of L-ADTs.} 
Let $L=(A,B,(Z,\sqsubseteq^Z,\sqcup^Z),z_0,\tau,\delta)$ be an L-ADT. 
As $\tau$ ``commutes'', the state reached after a sequence of 
transitions is order-independent.
Hence, we can define 
a natural, deterministic, interval-sequential specification of $L$, $\cS_L$, 
as the set of interval-sequential histories
$z_0,I_1,R_1,z_1,\ldots,I_m,R_m,z_{m}$
such that:
\begin{itemize}
\item  $\forall i=1,\ldots,m$, $z_i=\bigsqcup^Z_{a\in I_{i-1}}
\tau(a,z_{i-1})$, i.e., every state $z_i$ is a join of operations in
$I_{i-1}$ applied to $z_{i-1}$.  
\item $\forall i=1,\ldots,m$, $\forall r\in R_i$, $r=\delta(a,z_i)$,
where $a$ is the matching invocation operation for $r$, i.e., every response in $R_i$ is
the result of the associated operation applied to state $z_i$.
\end{itemize}

%

\section{Reconfigurable lattice agreement: definition}
\label{sec:definition}

We define a reconfigurable lattice $(\Lat,\sqsubseteq)$ as the product of
the state spaces of an \emph{object} L-ADT
$(A^{\cO},B^{\cO},(\cO,\sqsubseteq^{\cO},\sqcup^{\cO}),O_0,\tau^{\cO},\delta^{\cO})$
and a \emph{configuration} L-ADT
$(A^{\cC},B^{\cC},(\cC,\sqsubseteq^{\cC},\sqcup^{\cC}),C_0,\tau^{\cC},\delta^{\cC})$
(see Section~\ref{sec:ladtDef}). 
That is, $(\Lat,\sqsubseteq)=(\cO\times\cC,\sqsubseteq^{\cO}\times\sqsubseteq^{\cC})$ with
the product join operator $\sqcup=\sqcup^{\cO}\times\sqcup^{\cC}$.
Our main tool is  the reconfigurable lattice agreement, a generalization
of lattice agreement operating on $(\Lat,\sqsubseteq)$.
We say that $\Lat$ is the set of \emph{states}. 
For a state $u=(O,C)\in\Lat$, we use notations $u.O=O$ and $u.C=C$. 

\myparagraph{Failure model.}
%
%
When a client $p$ invokes $\textit{propose}((O,C))$, we say that $p$
\emph{proposes} object state~$O$ and configuration state $C$. 
We say that $p$ \emph{learns} an object state $O'$ and a configuration~$C'$ 
if its \textit{propose} invocation returns $(O',C')$.

%
%
We say that a configuration $C$ is \emph{potential}  
if there is a set $\{C_1,\ldots, C_k\}$ of proposed configurations such 
that~$C=C_0\sqcup^\cC(\bigsqcup^{\cC}_{i=1,\ldots,k}C_i)$ (with~$C_0$ the initial configuration).
A configuration $C$ is said to be \emph{superseded} as soon as a process learns 
a state $(*,C')$ with~$C\sqsubseteq^{\cC} C'$ and $C\neq C'$.
At any moment of time, a configuration is \emph{active} if it is a potential but not 
yet superseded configuration. Intuitively, some quorum of a configuration should
remain ``reachable'' as long as the configuration is active.
We say that a replica $r$ is \emph{active} when it is a member of 
an active configuration $C$, i.e., $r\in \textit{member}(C)$. 
A replica is \emph{correct} if, from some point on, it is forever
active and not failed.
A \emph{client} is correct if it does not fail while executing a
\textit{propose} operation.

A configuration $C$ is \emph{available} if some set of replicas in $\textit{quorums}(C)$
contains only correct processes.
In arguing liveness in this paper, we assume the following:
\begin{itemize}
\item {\bf Configuration availability.} Any potential configuration 
that is never superseded must be available.
\end{itemize}

Therefore, if a configuration is superseded by a strictly larger
(w.r.t. $\sqsubseteq^{\cC}$) one, then it does not have to be
available, i.e., we can safely remove some replicas from it for maintenance. 

\myparagraph{Liveness properties.}
In a constantly reconfigured system, we may not be able to ensure
liveness to all operations.
A slow client can be always behind the active configurations:
its set of estimated potential configurations can
always be found to constitute a superseded configuration. 
Therefore, for liveness, we assume that only finitely many
reconfigurations occur. Otherwise, only lock-freedom may be provided.

Therefore, to get a reconfigurable object, we replace the liveness
property of lattice agreement with the following one:

\begin{itemize}

\item {\bf Reconfigurable Liveness.} In executions with finitely many
  distinct proposed configurations, every \textit{propose} operation invoked by a correct
  client eventually returns.
\end {itemize}  

Thus, the desired liveness guarantees are ensured as long as only
finitely many distinct configurations are proposed.
However, the clients are free to perform infinitely many
\emph{object} updates without making any correct client starve.   

Formally, \emph{reconfigurable lattice agreement} defined on  
$(\Lat,\sqsubseteq)=(\cO\times\cC,\sqsubseteq^{\cO}\times\sqsubseteq^{\cC})$
satisfies the Validity and Consistency properties of
lattice agreement (see Section~\ref{sec:model}) and the Reconfigurable
Liveness property above.  

%
Furthermore, we  can only guarantee liveness to  clients assuming that, eventually, 
every correct system participant (client or replica) is informed of the 
currently active configuration. It boils down to ensuring that an eventually 
consistent reconfigurable memory is available to store the greatest learnt 
configuration.

For simplicity, we assume that a reliable broadcast primitive~\cite{rsdp-book2011} is
available, ensuring that (i) every broadcast message was previously broadcast, (ii)~if a correct
process broadcasts a message $m$, then it eventually delivers $m$, and
(iii)~every message delivered by a correct process is eventually
delivered by every correct process.     
Note that \emph{Configuration availability} implies that an active configuration 
is either available or sufficiently responsive to be superseded.   

\section{Reconfigurable lattice agreement: implementation}

\label{sec:implementation}

We now present our main technical result, a reconfigurable implementation of generalized lattice agreement. 
This algorithm will then be used to implement reconfigurable objects.

\myparagraph{Overview.}
The algorithm is specified by the pseudocode of Figure~\ref{fig:rgla2}. Note that we assume that all procedures
(including sub-calls to the \emph{updateState} procedure) are executed \emph{sequentially} until they terminate 
or get interrupted by the wait condition in line~\ref{line:rgla2:wait}.

In the algorithm, every process (client or server)~$p$ maintains a \emph{state} variable $v_p\in\Lat$ 
storing its local estimate of the greatest committed object ($v_p.O$) and configuration ($v_p.C$) states, 
initialized to the initial element of the lattice $(O_0,C_0)$. 
We say that a state is committed if a process broadcasted it in line~\ref{line:rgla2:bcast}. 
Note that all learnt states are committed (possibly indirectly by another process), 
but a process may fail before learning its committed state. 
Every process $p$ also maintains $T_p$, the set of \emph{active} \emph{input} configuration states, i.e., 
input configuration states that are not superseded by the committed state estimate~$v_p$.  
For the object lattice, processes stores in $\mathit{obj}_p$ the join of all known proposed objects states. 

To propose $\mathit{prop}$, client $p$ updates its local variables through the \emph{updateState} procedure 
using its input object and configuration states, $\mathit{prop}.O$ and $\mathit{prop}.C$ (line~\ref{line:rgla2:initState}).
Clients then enter a while loop where they send \emph{requests} associated with their current sequence number 
$\textit{seq}_p$ and containing the triplet $(v_p,\mathit{obj}_p,T_p)$, to all replicas from \emph{every possible join} 
of active base configurations and wait until either (1)~they get interrupted by discovering a greater committed configuration 
through the underlying reliable broadcast, or (2)~for each possible join of input configurations with the commit estimate configuration, 
a quorum of its replicas responded with messages of the type $\langle (\textit{resp}, \textit{seq}_p), (v,s_O, S_C) \rangle$, 
where~$(v,s_O, S_C)$ corresponds to the replica updated values of its triple $(v_p,\mathit{obj}_p,T_p)$ (lines~\ref{line:rgla2:send}--\ref{line:rgla2:wait}).

Whenever a process (client or replica) $p$ receives a new request, response or broadcast of the type 
$\langle \textit{msgType},(v,s_O,S_C)\rangle$, it updates its commit estimate and object candidate 
by joining its current values with the one received in the message. It also 
merges its set of input configurations $T_p$ with the received input configurations, 
but the values superseded by the updated commit estimate are trimmed off $T_p$ 
(lines~\ref{line:rgla2:newCommit}--\ref{line:rgla2:newConfs}).
For replicas, they also send a response containing the updated triplet $(v_p,\mathit{obj}_p,T_p)$
to the sender of the request (line~\ref{line:rgla2:sendback}). 

If responses from quorums of all queried configurations are received and no
response contained a \emph{new}, not yet known, input configuration or a greater object state, 
then the couple formed by $\mathit{obj}_p$ and the join of the commit estimate configuration with all input configurations 
$\bigsqcup^\cC (\{v.C\}\cup T_p)$ 
is broacasted and returned as the new learnt state (lines~\ref{line:rgla2:acceptTest}-\ref{line:rgla2:return}). 
Otherwise, clients proceed to a new round. 

To ensure wait-freedom, we integrate a helping mechanism simply consisting in  having clients adopt 
their committed state estimate (line~\ref{line:rgla2:adoptCommit}). But, to know when a committed state 
is great enough to be returned, clients must first complete a communication round without interference 
from reconfigurations (line~\ref{line:rgla2:setAdoptLB}). After such a round, the join of all known states, 
stored in $\mathit{learnLB}$, can safely be used as lower bound to return a committed state. 
We say that such configuration is pre-committed, hence all committed states have a pre-committed configuration.

\begin{figure}[!]
\hrule \vspace{1mm}
 {\small
\setcounter{linenumber}{0}
\begin{tabbing}
 bbbb\=bbbb\=bbbb\=bbbb\=bbbb\=bbbb\=bbbb\=bbbb \=  \kill
\textbf{Local variables:} \\
$\textit{seq}_p$, initially $0$  \commentline{The number of issued requests}\\
$v_p$, initially $(O_0,C_0)$  \commentline{The last learnt state}\\
$T_p$, initially $\emptyset$  \commentline{The set of proposed configuration states}\\
$\mathit{obj}_p$, initially $O_0$  \commentline{The candidate object state}\\
\\
\textbf{operation} \textit{propose}$(\mathit{prop})$
\commentline{Propose a new state $\mathit{prop}$}\\
\nnll\label{line:rgla2:initState}\> $\textit{updateState}(v_p,\mathit{prop}.O,\{\mathit{prop}.C\})$\\
\nnll\label{line:rglpassinga2:initAdoptLB}\> $\mathit{learnLB} := \bot$\\
\nnll\label{line:rgla2:while}\> \textbf{while} \textit{true} \textbf{do}\\
\nnll\label{line:rgla2:seqNum}\>\> $\textit{seq}_p := \textit{seq}_p+1$\\
\nnll\label{line:rgla2:refCommit}\>\> $\mathit{oldCommit} :=  v_p$ \commentline{Archive commit estimate}\\
\nnll\label{line:rgla2:refCandidates}\>\> $\mathit{oldCandidates} :=  (\mathit{obj}_p,T_p)$ \commentline{Archive candidate states}\\
\nnll\label{line:rgla2:queriedConfs}\>\> $V :=  \{\bigsqcup^\cC (\{v_p.C\}\cup S)~|~S\subseteq T_p\}$\commentline{Queried configurations}\\

\nnll\label{line:rgla2:send}\>\> send $\langle (\textit{REQ},
\textit{seq}_p),(v_p,\mathit{obj}_p,T_p)\rangle$ to $\bigcup_{u \in V} \textit{members}(u)$ \\ 
\nnll\label{line:rgla2:wait}\>\> \textbf{wait until} $\mathit{oldCommit}.C \neq v_p.C$ or 
$\forall u \in V$, received responses of the type  \\
\>\>\>\>\>  $\langle (\textit{RESP},\textit{seq}_p),\_ \rangle$ from some $Q\in\textit{quorums}(u)$\\
\nnll\label{line:rgla2:noInterrupt}\>\> \textbf{if}  $\mathit{oldCommit}.C=v_p.C\wedge \mathit{oldCandidates} = (\_,T_p)$ \textbf{then}  \commentline{Stable configurations}\\
\nnll\label{line:rgla2:setAdoptLB}\>\>\>  \textbf{if} $\mathit{learnLB}=\bot$ \textbf{then} $\mathit{learnLB}=(\mathit{obj}_p,\bigsqcup^\cC (\{v_p.C\}\cup T_p)$\\
\nnll\label{line:rgla2:acceptTest}\>\>\>  \textbf{if} $\mathit{oldCandidates} = (\mathit{obj}_p,\_)$ \textbf{then}  \commentline{No greater object received } \\
\nnll\label{line:rgla2:bcast}\>\>\>\> \textit{broadcast} $\langle \textit{COMMIT},((\mathit{obj}_p,\bigsqcup^\cC (\{v_p.C\}\cup T_p),\mathit{obj}_p,\emptyset) \rangle$\\
\nnll\label{line:rgla2:return}\>\>\>\> \textbf{return} $(\mathit{obj}_p,\bigsqcup^\cC (\{v_p.C\}\cup T_p))$\\ 
\nnll\label{line:rgla2:adoptCommit}\>\>  \textbf{if} $\mathit{learnLB}\neq\bot\wedge\mathit{learnLB}\sqsubseteq v_p$ \textbf{then return} $v_p$ \commentline{Adopt learnt state}\\
\\
\textbf{upon} \textit{receive} $\langle \textit{msgType},\textit{msgContent} \rangle$ from process $q$ \\
\nnll\label{line:rgla2:mergeReceived}\> $updateState(\textit{msgContent})$ \commentline{Update tracked states} \\
\nnll\label{line:rgla2:sendback}\> \textbf{if} $\textit{msgType}=(\textit{REQ},\textit{seq}) $ \textbf{then} send $\langle (\textit{RESP},\textit{seq}),(v_p,\mathit{obj}_p,T_p)\rangle$ to $q$\\ 
\\
\textbf{procedure} \textit{updateState}$(v,s_O,S_C)$ \commentline{Merge tracked states}\\
\nnll\label{line:rgla2:newCommit}\> $v_p := v_p\sqcup v$ \commentline{Update the commit estimate}\\
\nnll\label{line:rgla2:newObject}\> $\mathit{obj}_p := \mathit{obj}_p\sqcup^\cO s_O$ \commentline{Update the object candidate}\\
\nnll\label{line:rgla2:newConfs}\> $T_p := \{u\in (T_p\cup S_C)~|~u\not\sqsubseteq^\cC v_p.C \}$ \commentline{Update and trim input candidates}\\
\end{tabbing}
}
\vspace{-1.5mm}
 \hrule
\caption{Reconfigurable universal construction: code for process $p$.}

\label{fig:rgla2}
\end{figure}

\myparagraph{Correctness preliminaries.}
Let us first show that elements of the type $(v,s_O,S_C)\in \mathcal{L}\times\cO\times 2^\cC$, in which we 
have that~$\forall u\in S_C, u\not\sqsubseteq^\cC v.C$, admit a partial order $\sqsubseteq^*$ defined as:
\[
(v,s_O,S_C)\sqsubseteq^* (v',s_O',S_C') \Leftrightarrow v\sqsubseteq v' \wedge s_O\sqsubseteq^\cO s_O' \wedge 
\{u\in S_C| u\not\sqsubseteq^\cC v'.C\} \subseteq S_C'{}.
\]

\begin{property}
The relation $\sqsubseteq^*$ is a partial order on elements of the type $(v,s_O,S_C)\in \mathcal{L}\times\cO\times 2^\cC$, in which we 
have that~$\forall u\in S_C, u\not\sqsubseteq^\cC v.C$.
\end{property}

\begin{proof}
Note that, since $\sqsubseteq$ and $\sqsubseteq^\cO$ are partial orders, the reflexivity, symmetry and transitivity properties 
are verified if they are verified by the relation $\{u\in S_C| u\not\sqsubseteq^\cC v'.C\} \subseteq S_C'$.
The symmetry property is trivially verified as for any property $\mathcal{P}$, we have $\{u\in S_C| \mathcal{P}(u)\}\subseteq S_C$.
For transitivity, $(v,s_O,S_C)\sqsubseteq^* (v',s_O',S_C')$ and $(v',s_O',S_C')\sqsubseteq^* (v'',s_O'',S_C'')$ implies that: 
\[
\{u\in S_C| u\not\sqsubseteq^\cC v''.C\}\subseteq \{u\in \{w\in S_C|w\not\sqsubseteq^\cC v'.C\}|u\not\sqsubseteq^\cC v''.C\}\subseteq\{u\in S_C'|u\not\sqsubseteq^\cC v''.C\}\subseteq S_C''{}.
\]
Hence, that $(v,s_O,S_C)\sqsubseteq^*(v'',s_O'',S_C'')$. 
For antisymmetry, given $(v,s_O,S_C)\sqsubseteq^* (v',s_O',S_C')$ and $(v',s_O',S_C')\sqsubseteq^* (v,s_O,S_C)$, 
the relations $\sqsubseteq$ and $\sqsubseteq^\cC$ imply that $v=v'$ and $s_O=s_O'$. 
But as by assumption $\forall u\in S_C, u\not\sqsubseteq^\cC v.C$, we have $S_C= \{u\in S_C| u\not\sqsubseteq^\cC v.C\}$. 
Moreover, since~$v=v'$ we have $S_C= \{u\in S_C| u\not\sqsubseteq^\cC v'.C\}$ and we obtain that $S_C\subseteq S_C'$. 
Likewise, we have $S_C'\subseteq S_C$, and thus, we obtain that $S_C=S_C'$.
\end{proof}

Intuitively, the set of elements $(v,s_O,S_C)\in \mathcal{L}\times\cO\times 2^\cC$, in which we 
have that~$\forall u\in S_C, u\not\sqsubseteq^\cC v.C$, equipped with the partial order $\sqsubseteq^*$ is 
a join semi-lattice in which the procedure \emph{updateState} replaces the triple $(v_p,\mathit{obj}_p,T_p)$ 
with a join of itself and the procedure argument. But, we will only prove that the procedure \emph{updateState} 
replace $(v_p,\mathit{obj}_p,T_p)$ with an upper bound of itself and the procedure argument $(v,s_O,S_C)$:
 
\begin{lemma}
\label{lem:growingState}
Let $(v_p^\mathit{old},\mathit{obj}_p^\mathit{old},T_p^\mathit{old})$ and $(v_p^\mathit{new},\mathit{obj}_p^\mathit{new},T_p^\mathit{new})$
be the value of $(v_p,\mathit{obj}_p,T_p)$ respectively before and after an execution of the \emph{updateState}
procedure with argument $(v,s_O,S_C)$, then, we have:
\[
(v_p^\mathit{old},\mathit{obj}_p^\mathit{old},T_p^\mathit{old})\sqsubseteq^*(v_p^\mathit{new},\mathit{obj}_p^\mathit{new},T_p^\mathit{new})
\wedge (v,s_O,S_C) \sqsubseteq^*(v_p^\mathit{new},\mathit{obj}_p^\mathit{new},T_p^\mathit{new}){}.
\]
\end{lemma}

\begin{proof}
Let us first note that we can rewrite the operation as follows: 
\begin{itemize}
\item Line~\ref{line:rgla2:newCommit}: $v_p^\mathit{new}= v_p^\mathit{old}\sqcup v$
\item Line~\ref{line:rgla2:newObject}: $\mathit{obj}_p^\mathit{new}= \mathit{obj}_p^\mathit{old}\sqcup^\cO s_O$
\item Line~\ref{line:rgla2:newConfs}: $T_p^\mathit{new}= \{u\in (T_p^\mathit{old}\cup S_C)| u\not\sqsubseteq^\cC (v_p^\mathit{old}\sqcup^\cC v).C\}$
\end{itemize}
Hence, the use of $(v_p^\mathit{old},\mathit{obj}_p^\mathit{old},T_p^\mathit{old})$ and $(v,s_O,S_C)$ are symmetrical.
Moreover, it is trivial to check that, w.l.o.g., $(v,s_O,S_C)\sqsubseteq^*(v_p^\mathit{new},\mathit{obj}_p^\mathit{new},T_p^\mathit{new})$.
Indeed, $v\sqsubseteq v_p^\mathit{old}\sqcup v$, $s_O\sqsubseteq^\cO \mathit{obj}_p^\mathit{old}\sqcup^\cO s_O$ and 
$\{u\in S_C|u\not\sqsubseteq^\cC v_p^\mathit{new}.C\} \subseteq \{u\in (T_p^\mathit{old}\cup S_C)| u\not\sqsubseteq^\cC (v_p^\mathit{old}\sqcup^\cC v).C\} = T_p^\mathit{new}$.
\end{proof}

Note that it is also trivial to check that initially we have $\forall u\in T_p, u\not\sqsubseteq^\cC v_p.C$ as $T_p=\emptyset$ 
and that it remains true after a complete execution of the \emph{updateState} procedure as $T_p$ is 
taken as the set of elements of $(T_p^\mathit{old}\cup S_C)$ satisfying this condition.

Let us now check that $\sqsubseteq^*$ is a refinement of the order $\sqsubseteq$ for the projection $\mathit{decide}()$ 
defined such that $\mathit{decide}(v,s_O,S_C)=(s_O,\bigsqcup^\cC(\{v.C\}\cup S_C))$. Formally:

\begin{lemma}
$(v,s_O,S_C)\sqsubseteq^*(v',s_O',S_C') \implies \mathit{decide}(v,s_O,S_C)\sqsubseteq\mathit{decide}(v',s_O',S_C')$.
\label{lem:orderProjection}
\end{lemma}

\begin{proof}
This result follows directly from the definition of $\sqsubseteq^*$. Indeed, as $(v,s_O,S_C)\sqsubseteq^*(v',s_O',S_C')$, 
we have $s_O\sqsubseteq^\cO s_O'$. Moreover, we have $\{u\in S_C| u\not\sqsubset^\cC v'.C\}\subseteq S_C'$. Hence 
we have $\bigsqcup^\cC(\{v'.C\}\cup S_C)\sqsubseteq^\cC \bigsqcup^\cC(\{v'.C\}\cup S_C')$. But, as moreover we have 
$v\sqsubseteq v'$, we obtain that $\bigsqcup^\cC(\{v.C\}\cup S_C)\sqsubseteq^\cC \bigsqcup^\cC(\{v'.C\}\cup S_C')$.
\end{proof}

\myparagraph{Consistency.}
Let us start with the consistency proof. For this, consider any run of the algorithm in Figure~\ref{fig:rgla2}.
Let $s$ be any state committed in the considered run.
Let $p(s)$ denote \emph{the first} client that committed $s$ in line~\ref{line:rgla2:bcast}. 
Let $V(s)$, $v(s)$, $\mathit{obj}(s)$ and $T(s)$ denote the value of respectively the variables $V$, $v_{p(s)}$, 
$\mathit{obj}_{p(s)}$ and $T_{p(s)}$ at the moment when $p(s)$ commited $s$ in line~\ref{line:rgla2:bcast}.
Note that, as~$p(s)$ passed the tests in lines~\ref{line:rgla2:noInterrupt} and~\ref{line:rgla2:acceptTest}, 
$v_{p(s)}.C$, $\mathit{obj}_{p(s)}$ and $T_{p(s)}$ must have remained unchanged and equal to respectively 
$v(s).C$, $\mathit{obj}(s)$ and $T(s)$ since the last computation of~$V$ in line~\ref{line:rgla2:queriedConfs}. 
In particular, we have $V(s)=\{\bigsqcup^\cC (\{v(s).C\}\cup S)~|~S\subseteq T(s)\}$.

Let $G_s$ be the graph whose vertices are all the committed states plus the initial state~$(O_0,C_0)$ 
and whose edges are defined as follows: 
\[
s\rightarrow s' \; \Leftrightarrow\; s\sqsubsetneq s' \wedge  s.C \in V(s'). 
\]  
Let us first show that some general observation about $G_s$, that is:
\begin{lemma}
    \label{lem:stateParents}
For any committed configuration state $s$, we have $v(s) \rightarrow s$.
\end{lemma}
\begin{proof}
Let $s$ be any committed configuration, we have $v(s).C\in V(s)$ as~$v(s).C$ is 
the value of $v_{p(s)}.C$ used in the computation of $V(s)$ in line~\ref{line:rgla2:queriedConfs}. 
Hence, as $v(s)\sqsubseteq s$ since $s=(\mathit{obj}(s),\bigsqcup^\cC (\{v(s).C\}\cup T(s))$ and as $v(s)\neq s$ 
since $p(s)$ is the first process to commit $s$, we obtain that $v(s)\sqsubsetneq s$.
\end{proof}

Note that it implies that $G_s$ admits a single source $(O_0,C_0)$. Moreover, it is acyclic as $\sqsubsetneq$ is a partial order.

Let us now show the main result concerning $G_s$ derived from the algorithm, that is:
  \begin{lemma}
    \label{lem:stateForks}
Given $\bar s$, $s$ and $s'$ in $G_s$ if $\bar s\rightarrow s$, $\bar s\rightarrow s'$, $v(s')\sqsubseteq s$ and $v(s)\sqsubseteq s$ then either~$s\rightarrow s'$ or else $s'\rightarrow s$.
  \end{lemma}
  \begin{proof}
  Let us consider $\bar s$, $s$ and $s'$ in $G$ such that $\bar s\rightarrow s$, $\bar s\rightarrow s'$, $v(s')\sqsubseteq s$ and $v(s)\sqsubseteq s$. From~$\bar s\rightarrow s$ and $\bar s\rightarrow s'$, we can derive that:
\[\bar s.C \in V(s)\wedge \bar s.C\in V(s')\implies \bar s.C \in V(s)\cap V(s'){}.\]
Let us now look back at the algorithm to show that an edge must exist from $s$ to $s'$ or from~$s'$ to~$s$.
By the algorithm, as $\bar s.C \in V(s)\cap V(s')$, in the last round of requests before committing~$s$ (resp. $s'$), 
$p(s)$ (resp. $p(s')$) sent a request to all processes in $\bar s.C$. 
As, in their last round, $p(s)$ and $p(s')$ passed the test of line~\ref{line:rgla2:wait}, 
they received responses from replicas of~$\bar s.C$ forming \emph{quorums} in $\bar s.C$, 
hence, as quorums intersect, from a common process~$r\in \bar s.C$.

Let us first assume that, w.l.o.g., for their last round of requests, $r$ responded to $p(s)$ before responding to $p(s')$. 
Recall that, as $p(s)$ passed the tests in lines~\ref{line:rgla2:noInterrupt} and~\ref{line:rgla2:acceptTest}, 
the values of $v_{p(s)}.C$, $\mathit{obj}_{p(s)}$ and~$T_{p(s)}$ did not change in the last round.  
Hence the content of the request sent to $r$ by $p(s)$ is equal to $((v_O,v(s).C),\mathit{obj}(s),T(s))$, 
with $v_O$ some arbitrary value. By Lemma~\ref{lem:growingState}, 
after $r$ responded to $p(s)$, $(v_r,\mathit{obj}_r,T_r)$ must become and remain greater or equal to (w.r.t. $\sqsubseteq^*$)
the message content $((v_O,v(s).C),\mathit{obj}(s),T(s))$. Hence, 
the latter response to $p(s')$ by $r$ must contain a greater or equal content, 
and $(v_{p(s')},\mathit{obj}_{p(s')},T_{p(s')})$ becomes and remains 
greater or equal to~$((v_O,v(s).C),\mathit{obj}(s),T(s))$, thus 
$((v_O,v(s).C),\mathit{obj}(s),T(s))\sqsubseteq^*(v(s'),\mathit{obj}(s'),T(s'))$. 

By applying the result of Lemma~\ref{lem:orderProjection}, we get that $s = \mathit{decide}((v_O,v(s).C),\mathit{obj}(s),T(s)) \sqsubseteq^* \mathit{decide}(v(s'),\mathit{obj}(s'),T(s'))= s'$, so that $s\sqsubsetneq s'$.

Let us now conclude by showing that we also have $s.C\in V(s')$.
As $v(s')$ did not change during the round, it must be greater than $v(s)$. Moreover, by assumption it is smaller than $s$, 
hence we have~$v(s)\sqsubseteq^* v(s')\sqsubseteq^* s$. Thus:
\[
s.C= \bigsqcup^\cC(\{v(s).C\}\cup T(s)) \sqsubseteq^\cC \bigsqcup^\cC (\{v(s').C\}\cup T(s))\sqsubseteq^\cC \bigsqcup^\cC (\{s.C\}\cup T(s))= s.C{}.
\]
So $s.C = \bigsqcup^\cC (\{v(s').C\}\cup T(s))$, and hence, $s.C = \bigsqcup^\cC (\{v(s').C\}\cup \{u\in T(s), u \not \sqsubseteq^\cC v(s').C\})$. 
From $((v_O,v(s).C),\mathit{obj}(s),T(s))\sqsubseteq^*(v(s'),\mathit{obj}(s'),T(s'))$, we get $\{u\in T(s), u \not \sqsubseteq^\cC v(s').C\}\subseteq T(s')$, and therefore, we obtain that:
\[
s.C=\bigsqcup^\cC (\{v(s').C\}\cup \{u\in T(s), u \not \sqsubseteq^\cC v(s').C\})\in \{\bigsqcup^\cC (\{v(s').C\}\cup S)~|~S\subseteq T(s')\} = V(s'){}.
\]
Hence $s.C\in V(s')$ and thus there is an edge from $s$ to $s'$ in $G_s$.
  \end{proof}

Let us now show that $G_s$ is a connected graph:
\begin{lemma}
	\label{lem:connected}
    $G_s$ is connected.
\end{lemma}
\begin{proof}
Let us show this result by contradiction.
Hence, let us assume that we can select committed states $s$ and $s'$, 
such that $(s,s')$ is a \emph{minimal} (w.r.t.~$\sqsubseteq$) pair of vertices of~$G_s$ 
that are not connected via a path.

Let us first show that $s$ and $s'$ share the same set of ancestors in $G_s$.
Indeed, consider an ancestor $u$ of $s$ in $G_s$. As $u \sqsubseteq s$ 
and as~$(s,s')$ is chosen minimal, there exists a path from $u$ to~$s'$ or from $s'$ to $u$. 
There is no path from~$s'$ to~$u$ as it would imply a path from~$s'$ to~$s$.
Hence,~$u$ is an ancestor of~$s'$. By symmetry between $s$ and $s'$,
we get that $s$ and~$s'$ share the same set of ancestors in $G_s$. 

All ancestors being connected, they are totally ordered by $\sqsubseteq$. 
Hence, let $\bar s$ be the maximal ancestor of $s$ and $s'$. 
The paths from $\bar s$ to $s$ and $s'$ must be edges as $\bar s$ is the greatest common ancestor. 
Moreover, $v(s)$ and $v(s')$ are ancestors of $s$ and $s'$ and therefore
we have $v(s')\sqsubseteq s$ and $v(s)\sqsubseteq s'$
Thus, we can apply Lemma~\ref{lem:stateForks} to obtain that there is 
and edge, thus a path between $s$ and $s'$ --- a contradiction.
\end{proof}

\begin{theorem}
\label{thm:consistency}
The algorithm in Figure~\ref{fig:rgla2} satisfies the consistency property.
\end{theorem}
\begin{proof}
For the \textbf{Consistency} property, Lemma~\ref{lem:connected} says that $G$ is connected, 
and hence that all committed states are totally ordered, thus, 
that all learnt states are totally ordered. 
\end{proof}

\myparagraph{Validity.}
The proof of validity is very similar to the proof of consistency. 
Consider any run of the algorithm in Figure~\ref{fig:rgla2}.
Let $c$ be any configuration state pre-committed in the considered run.
Let $p(c)$ denote \emph{the first} client that pre-committed configuration $c$ in line~\ref{line:rgla2:noInterrupt}. 
Let $V(c)$, $v(s)$, $\mathit{obj}(c)$ and $T(c)$ denote the value of respectively the variables $V$, $v_{p(c)}$, 
$\mathit{obj}_{p(c)}$ and $T_{p(c)}$ when $p(c)$ pre-commited $c$ in line~\ref{line:rgla2:noInterrupt}.
%
Note that, as~$p(c)$ passed the test in line~\ref{line:rgla2:noInterrupt}, 
$v_{p(c)}.C$ and~$T_{p(c)}$ must have remained unchanged and equal to respectively~$v(c).C$ 
and~$T(c)$ since the last computation of~$V$ in line~\ref{line:rgla2:queriedConfs}. 
In particular, we have $V(c)=\{\bigsqcup^\cC (\{v(c).C\}\cup S)~|~S\subseteq T(c)\}$.

We say that a configuration $c$ becomes inactive, at time $c.\mathit{inactive}$, when for any replica $r$ 
of one of its quorum reach a state with a greater configuration state 
(i.e., $\exists Q\in \mathit{quorum}(c), \forall r\in Q, c\sqsubsetneq^\cC\bigsqcup^\cC(\{v_r.C\}\cup T_r)$).
We also consider, if any, the time of the commit of the corresponding operation. This time, denoted $c.\mathit{commit}$, 
correspond to the first commit with this configuration, if any, and is equal to $+\infty$ otherwise.

Let $G_c$ be the graph whose vertices are all pre-committed states plus the initial configuration~$C_0$ 
and whose edges are defined as follows: 
\[
c\rightarrow c' \; \Leftrightarrow\; c\sqsubsetneq^\cC c' \wedge  c \in V(c') \wedge c.\mathit{inactive} <  c'.\mathit{commit}. 
\]  
Let us first show that some genreal observation $G_c$, that is:
\begin{lemma}
    \label{lem:commitParent}
For any $c\in G_c$, we have $v(c).C \rightarrow c$ and $(v(c).C).\mathit{commit} <  c.\mathit{commit}$.
\end{lemma}
\begin{proof}
Let $c$ be any pre-committed configuration in $G_c$, we have $v(c).C\in V(c)$ as~$v(c).C$ 
is the value of $v_{p(c)}.C$ used in the computation of $V(c)$ in line~\ref{line:rgla2:queriedConfs}. 
Hence, as $v(c).C\sqsubseteq^\cC c$ since~$c=\bigsqcup^\cC (\{v(c).C\}\cup T(c))$ and as $v(c).C\neq c$ 
since $p(c)$ is the first process to pre-commit $c$, we obtain that $v(c)\sqsubsetneq^\cC c$. 
Moreover, a quorum of $v(c).C$ responded to the request made by $p(c)$ in its last round with $c$. 
So as $v(c)\sqsubsetneq^\cC c$, we have $(v(c).C).\mathit{inactive}<c.\mathit{commit}$. 
Additionnaly, a committed configuration can be adopted only after the first operation committing it terminated. 
Therefore, we also have $(v(c).C).\mathit{commit} <  c.\mathit{commit}$.
\end{proof}

Let us now show the main result about $G_c$ derived from the algorithm, that is:
  \begin{lemma}
    \label{lem:commitForks}
$\forall c,c' \in G_c: v(c).C\rightarrow c'\wedge v(c').C\sqsubseteq c \implies c\rightarrow c' \vee c'\rightarrow c$.
  \end{lemma}
  \begin{proof}
  Let us consider $c$ and $c'$ in $G$ such that $v(c).C\rightarrow c'$ and $v(c')\sqsubseteq c$. Note 
  that by Lemma~\ref{lem:commitParent}, we also have  $v(c).C\rightarrow c$, hence in particular we have:
\[v(c).C \in V(c)\wedge v(c).C\in V(c')\implies v(c).C \in V(c)\cap V(c'){}.\]
Let us now look back at the algorithm to show that an edge must exist from $c$ to $c'$ or from~$c'$ to~$c$.
By the algorithm, as $v(c).C \in V(c)\cap V(c')$, in the last round of requests before committing~$c$ 
(resp. $c'$), $p(c)$ (resp. $p(c')$) sent a request to all processes in $v(c).C$. 
As, in their last round, $p(c)$ and $p(c')$ passed the test of line~\ref{line:rgla2:wait}, 
they received responses from replicas of~$v(c).C$ forming \emph{quorums} in $v(c).C$, 
hence, as quorums intersect, from a common process~$r\in v(c).C$.

Let us first assume that, w.l.o.g. (note that we also have $v(c).C\sqsubseteq^\cC c'$ as $v(c).C\rightarrow c'$), 
for their last round of requests, $r$ responded to $p(c)$ before responding to $p(c')$. 
Note that this already implies that $c.\mathit{inactive} <  c'.\mathit{commit}$.
Recall that, as $p(c)$ passed the test in line~\ref{line:rgla2:noInterrupt}, 
the values of $v_{p(c)}.C$ and~$T_{p(c)}$ did not change in the last round.  
Hence the content of the request sent to $r$ by $p(c)$ is equal to $((v_O,v(c).C),v_O',T(c))$, 
with $v_O, v_O'$ some arbitrary values. By Lemma~\ref{lem:growingState}, 
after $r$ responded to $p(c)$, $(v_r,\mathit{obj}_r,T_r)$ must become and remain greater or equal to (w.r.t. $\sqsubseteq^*$)
the message content $((v_O,v(c).C),v_O',T(c))$. Hence, 
the latter response to $p(c')$ by $r$ must contain a greater or equal content, 
and $(v_{p(c')},\mathit{obj}_{p(c')},T_{p(sc')})$ becomes and remains 
greater or equal to~$((v_O,v(c).C),v_O',T(c))$, thus 
$((v_O,v(c).C),v_O',T(c))\sqsubseteq^*(v(c'),\mathit{obj}(c'),T(c'))$. 

By applying the result of Lemma~\ref{lem:orderProjection}, we get $(v_O',c) = \mathit{decide}((v_O,v(c).C),v_O',T(c)) \sqsubseteq^* \mathit{decide}(v(c'),\mathit{obj}(c'),T(c'))= (\mathit{obj}(c'),c')$, so that $c\sqsubseteq^\cC c'$, hence, $c\sqsubsetneq^\cC c'$.

Let us now conclude by showing that we also have $c\in V(c')$.
As $v(c')$ did not change during the round, it must be greater than $v(c)$. Moreover, by assumption it is smaller than $c$, 
hence we have~$v(c).C\sqsubseteq^\cC v(c').C\sqsubseteq^\cC c$. Thus:
\[
c= \bigsqcup^\cC(\{v(c).C\}\cup T(c)) \sqsubseteq^\cC \bigsqcup^\cC (\{v(c').C\}\cup T(c))\sqsubseteq^\cC \bigsqcup^\cC (\{c\}\cup T(c))= c{}.
\]
So $c = \bigsqcup^\cC (\{v(c').C\}\cup T(c))$, and hence, $c = \bigsqcup^\cC (\{v(c').C\}\cup \{u\in T(c), u \not \sqsubseteq^\cC v(c').C\})$. 
From $((v_O,v(c).C),\mathit{oldObj}(c),T(c))\sqsubseteq^*(v(c'),\mathit{obj}(c'),T(c'))$, we get that $\{u\in T(c), u \not \sqsubseteq^\cC v(c').C\}\subseteq T(c')$, and therefore, we obtain that:
\[
c=\bigsqcup^\cC (\{v(c').C\}\cup \{u\in T(c), u \not \sqsubseteq^\cC v(c').C\})\in \{\bigsqcup^\cC (\{v(c').C\}\cup S)~|~S\subseteq T(c')\} = V(c'){}.
\]
Hence $c\in V(c')$ and thus there is an edge from $c$ to $c'$ in $G_c$.
  \end{proof}

The validity proof differs from the consistency one by showing that $G_c$ is connected through specific paths.
We say that the \emph{commit parent} of any $c\in G_c$ is $v(c).C$. We say that there is a commit path from $c$ to $c'$, 
denoted as $c \leadsto_c c'$, if there is a a sequence of commit parents from~$c'$, that is $c_1',\dots, c_k'=c'$ with 
$\forall i\in \{2,\dots,k\}: c_{i-1}'=v(c_i').C$, with $c \rightarrow c_1'$.

\begin{lemma}
\label{lem:commitConnected}
$\forall c,c' \in G_c: c\leadsto_c c' \vee c'\leadsto_c c$.
\end{lemma}

\begin{proof}
Let us show that given two sequences of commit parents $c_1,\dots, c_k$ and $c_1',\dots, c_l'$ such 
that $c_1\rightarrow c_1'$, $k\neq 1$ and $v(c_1')\sqsubseteq^\cC c_2$ then we have sequences of commit parents 
$c_2,\dots, c_k$ and $c_1',\dots, c_l'$ with either $c_1' \rightarrow c_2$ or else $c_2 \rightarrow c_1'$. 
Note that this is a direct application of Lemma~\ref{lem:commitForks}. Moreover, in the former case, 
if $l\neq 1$, then we have $c_1\rightarrow c_1'\rightarrow c_2'$ and therefore~$v(c_2)=c_1\sqsubseteq c_2'$. 
In the latter case, if $k\neq 2$, then we have $v(c_1')\sqsubseteq^\cC c_3$ as $v(c_1')\sqsubseteq^\cC c_2$ and 
$c_2\sqsubseteq^\cC c_3$ from $c_2\rightarrow c_3$.

Therefore, given such sequences, we can apply this property inductively until we consume all of one sequence. 
Hence, we obtain either $c_k\leadsto_c c_l'$ or $c_l'\leadsto_c c_k$. 

Now, consider any two pre-committed configurations $c,c'$ in $G_c$. 
By a trivial recursive application of Lemma~\ref{lem:commitParent} to $c$ and $c'$, we obtain that 
there exists sequences of commit parents from $C_0$ to both $c$ and $c'$. Let  $C_0=c_1,\dots, c_k=c$ 
and $C_O = c_1',\dots, c_l'=c'$ be these sequences. Note that both $k$ and $l$ are not equal to $1$ as 
$c$ and $c'$ are distinct fron $C_0$. Moreover, by applying Lemma~\ref{lem:commitForks}, we obtain that either 
$c_1\rightarrow c_1'$ or $c_1'\rightarrow c_1$. Moreover, we have $C_0\sqsubseteq^\cC c_2$ and 
$C_0\sqsubseteq^\cC c_2'$. Everything is in place to apply our inductive result and obtain that either 
we have $c\leadsto_c c'$ or else we have $c'\leadsto_c c$.
\end{proof}

We now have all the ingredients to show that Algorithm~\ref{fig:rgla2} satisfies the validity property.

\begin{theorem}
\label{thm:validity}
The algorithm in Figure~\ref{fig:rgla2} satisfies the consistency property.
\end{theorem}
\begin{proof}

A learnt state include the operation input as states only increase as shown by Lemma~\ref{lem:growingState} and lemma~\ref{lem:orderProjection}.
Thus, let us show that all preceding committed states, hence learnt, are included in any learn state.

Let us first show that a pre-commit by an operation implies that 
this operation did not start after the associated configuration became inactive. 
Indeed, if the configuration $c$ is already inactive, then an operation querying $c$ 
must return a greater configuration. Hence, it cannot pre-commit it --- a contradiction.
 
But, as shown in Lemma~\ref{lem:commitConnected}, there must exist 
a commit path between a pre-committed state and any other committed state. 
Moreover, as shown in Lemma~\ref{lem:commitParent}, connected parents must have committed before. 
Hence, any committed state with a greater configuration must have been 
committed after the current configuration became inactive. 
Consequently, any commit with a greater configuration cannot precede the ongoing operation.

Moreover, at the time a pre-commit happens, all preceding learnt state must have reached a quorum of the current configuration. 
By Lemma~\ref{lem:growingState}, this quorum must have become and remained greater than any of these preceding learnt states. 

Hence, a decision based on a pre-commit would include all preceding learnt states. 
It is enough to show validity as either a learnt state comes directly from 
a pre-commit (in Line~\ref{line:rgla2:newCommit}) or is greater than the value 
stored in $\mathit{LearnLB}$ at the time of a pre-commit.
\end{proof}

\myparagraph{Reconfigurable-liveness.} Let us directly show that we have reconfigurable-liveness:
\begin{theorem}
\label{thm:liveness}
The algorithm in Figure~\ref{fig:rgla2} satisfies the reconfigurable-liveness property.
\end{theorem}
\begin{proof}
To prove the \textbf{Reconfigurable-Liveness} property, consider a run
in which only finitely many distinct \emph{configurations} are
proposed. Hence, there exists a greatest learnt configuration state $C_f$. 
By the properties of the reliable-broadcast mechanism (line~\ref{line:rgla2:bcast}), 
eventually, all correct processes will receive a commit message including $C_f$.
Hence, eventually, all correct processes will have $v_p.C=C_f$. 

Assuming \textbf{configuration availability}, we have that every
join of proposed configurations that are not yet superseded must have an
available quorum. Thus, eventually, every configurations $u.C$ queried by 
correct processes are available.
Therefore, correct processes cannot be blocked forever waiting in
line~\ref{line:rgla2:wait} and, thus, they have to perform infinitely many
iterations of the while loop. 
Moreover, since eventually no new configuration is discovered, all 
correct processes will eventually always pass the test in line~\ref{line:rgla2:noInterrupt}
and therefore set a state for $\mathit{learnLB}$. 
In a round of requests after setting $\mathit{learnLB}$ based on the triple $(v_l,\mathit{obj}_l,T_l)$, 
the triple $(v_r,\mathit{obj}_r,T_r)$ in all replicas from a quorum of $C_f$ must become and remain greater 
(w.r.t~$\sqsubseteq^*$) than  $(v_l,\mathit{obj}_l,T_l)$.

Now, let us assume that a correct process $p$ never terminates, thus, 
it must observe greater object candidates at each round. 
It implies that infinitely many \emph{propose} procedures are initiated, 
hence that a process commits infinitely may states. 
A committed state must be computed based on a triple $(v_p,\mathit{obj}_p,T_p)$ 
greater than those in all received messages, 
in particular those from a quorum in $C_f$ which must eventually be greater than $(v_l,\mathit{obj}_l,T_l)$. 
Hence, eventually, a committed state greater than $\mathit{learnLB}$ is broadcasted, 
and this state is adopted and returned by $p$ after receiving it --- a contradiction.
 \end{proof}

Using Theorems~\ref{thm:consistency},~\ref{thm:validity} and~\ref{thm:liveness}, we obtain that:

\begin{theorem}
The algorithm in Figure~\ref{fig:rgla2} implements reconfigurable lattice agreement.
\end{theorem}

\section{Reconfigurable objects}

\label{sec:ladt}

In this section, we use our reconfigurable lattice agreement (RLA) abstraction
to construct an interval-linearizable reconfigurable implementation of
any L-ADT $L$.

\subsection{Defining and implementing reconfigurable L-ADTs} 
Let us consider two L-ADTs, an \emph{object} L-ADT 
$L^{\cO}=(A^{\cO},B^{\cO},(\cO,\sqsubseteq^{\cO},\sqcup^{\cO}),O_0,\tau^{\cO},\delta^{\cO})$
and a \emph{configuration} L-ADT
$L^{\cC}=(A^{\cC},B^{\cC},(\cC,\sqsubseteq^{\cC},\sqcup^{\cC}),C_0,\tau^{\cC},\delta^{\cC})$
(Section~\ref{sec:model}).

The corresponding \emph{reconfigurable L-ADT} implementation, defined on the composition
$L=L^{\cO}\times L^{\cC}$, exports operations in $A^{\cO}\times A^{\cC}$.
%
It must be interval-linearizable (respectively to $\cS_L$) and ensure
Reconfigurable Liveness (under the configuration availability
assumption).

In the reconfigurable implementation of $L$,
presented in Figure~\ref{Alg:ladt}, whenever a process invokes an
operation $a\in A^{\cO}$,
it proposes a state, $\tau^\cO(O_p,a)$---the result from applying $a$ to the last
learnt state (initially, $C_0$)---to RLA, updates $(O_p,C_p)$ and returns the response~$\delta^\cO(O_p,a)$
corresponding to the new learnt state.
Similarly, to update the configuration, the process applies its
operation to the last learnt configuration and proposes the
resulting state to RLA. 

\begin{figure}
\hrule \vspace{1mm}
 {\small
\setcounter{linenumber}{0}
\begin{tabbing}
 bbbb\=bbbb\=bbbb\=bbbb\=bbbb\=bbbb\=bbbb\=bbbb \=  \kill
\textbf{Shared:}  $\textit{RLA}$, reconfigurable lattice agreement\\
\textbf{Local variables:}\\
\> $O_p$, initially $O_0$  \commentline{The last learnt object state}\\
\> $C_p$, initially $C_0$  \commentline{The last learnt configuration state}\\
\textbf{upon invocation of} $a \in A^{\cO}$ \commentline{Object operation}\\
\nnll\label{line:ladt:invoke}\> $(O_p,C_p):=\textit{RLA}.\textit{propose}((\tau^{\cO}(O_p,a),C_p))$\\
\nnll\label{line:ladt:return}\> \textbf{return} $\delta^{\cO}(O_p,a)$\\
\textbf{upon invocation of} $a\in A^{\cC}$ \commentline{Reconfiguration}\\
\nnll\label{line:ladt:rinvoke}\> $(O_p,C_p):=\textit{RLA}.\textit{propose}((O_p, \tau^{\cC}(C_p,a)))$\\
\nnll\label{line:ladt:rreturn}\> \textbf{return} $\delta^{\cC}(C_p,a)$
\end{tabbing}
}
\vspace{-1.5mm}
 \hrule
\caption{Interval-linearizable implementation of L-ADT
 $L=L^{\cO}\times L^{\cC}$: code for process $p$.}

\label{Alg:ladt}
\end{figure}

\begin{theorem}
\label{th:ladt}
The algorithm in Figure~\ref{Alg:ladt} is a reconfigurable 
implementation of an L-ADT.
\end{theorem}
\begin{proof}
Consider any execution of the algorithm in Figure~\ref{Alg:ladt}.
 
By the Validity and Consistency properties of the underlying RLA
abstraction, we can represent the states and operations of the
execution as a sequence $z_0,I_1,z_1,\ldots,I_m,z_m$, where
$\{z_1,\ldots,z_m\}$ is the set of learnt values, and each $I_i$, $i=1,\ldots,m$, is a
set of operations invoked in this execution, such that
$z_i=\bigsqcup_{a\in I_i}\tau(a,z_{i-1})$.     

A construction of the corresponding interval-sequential
history is immediate. Consider an operation $a$ that returned a value in
the execution based on a learnt state $z_i$
(line~\ref{line:ladt:return}).
Validity of RLA implies that $a\in I_j$ for some $j\leq i$.
Thus, we can simply add $a$ to set $R_i$.
By repeating this procedure for every complete operation, we get a
history $z_0,I_1,R_1,z_1,\ldots,I_m,R_m,z_m$ complying with $\cS_L$.  
By construction, the history also preserves the precedence relation of the original history.

Reconfigurable liveness of the implementation is implied by the properties of RLA (assuming
reconfiguration availability).
\end{proof}

%
%
%
In the special case, when the L-ADT is \emph{update-query}, the
construction above produces a \emph{linearizable} implementation:

\begin{theorem}
\label{th:crdt}
The algorithm in Figure~\ref{Alg:ladt} is a reconfigurable linearizable
implementation of an update-query L-ADT.
\end{theorem}
\begin{proof}
Consider any execution of the algorithm in Figure~\ref{Alg:ladt} and
assume that $L$ is update-query.

By Theorem~\ref{th:ladt}, there exists a history
$z_0,I_1,R_1,z_1,\ldots,I_m,R_m,z_m$ that complies with $\cS_L$, the
interval-sequential specification of $L$.
We now construct a \emph{sequential} history satisfying the
\emph{sequential} specification of $L$ as follows:  

\begin{itemize}

\item For every update $u$ in the history, we match it with
immediately succeeding matching response $\bot$ (remove the other response of
$u$ if any);  

\item For every response of a query $q$ in the history we match it with an
immediately preceding matching invocation of $q$ (remove the other 
invocation of $q$ if any);   

\end{itemize}

As the updates of an L-ADT are commutative, the order in which we
place them in the constructed sequential history $S$ does not matter, and
it is immediate that every response in $S$
complies with $\tau$ and $\delta$ in a sequential history of $L$.
\end{proof}

\subsection{L-ADT examples}

We provide four examples of L-ADTs that allow for interval-linearizable
(Theorem~\ref{th:ladt}) and linearizable (Theorem~\ref{th:crdt})
reconfigurable implementations. 

\myparagraph{Max-register.}
The \textit{max-register} sequential object defined on a totally ordered set
$(V,\leq_V)$ provides  operations $\textit{writeMax}(v)$, $v\in V$,
returning a default value $\bot$, 
and  $readMax$  returning the largest value written so far (or
$\bot$ if there are no preceding writes). 
We can define the type as an update-query L-ADT as follows:
\[
	\textit{MR}_V=(\textit{writeMax}(v)_{v\in V}\cup \{\textit{readMax}\}, V\cup\{\bot\}, (V\cup\{\bot\},\leq_V,\textit{max}_V),\bot,\tau_{\textit{MR}_V},\delta_{\textit{MR}_V}){}.
\]
where $\leq_V$ is extended to $\bot$ with $\forall v\in V: \bot\leq_V v$, 
$\delta_{\textit{MR}_V}(z,a)=z$ if $a=\textit{readMax}$ and $\bot$ otherwise, 
and $\tau_{\textit{MR}_V}(z,a)=\textit{max}_V(z,v)$ if $a=\textit{writeMax}(v)$ and~$z$ otherwise.

It is easy to see that  $(V\cup\{\bot\},\leq_V,\textit{max}_V)$ is a join
semi-lattice and the L-ADT $\textit{MR}_V$ satisfies the sequential \textit{max-register} specification.
%

\myparagraph{Set.}
The (add-only) \textit{set} sequential object defined using a countable set $V$ 
provides operations $\textit{addSet}(v)$, $v\in V$, returning a default value $\bot$, 
and $readSet$ returning the set of all values added so far (or
$\emptyset$ if there are no preceding add operation). 
We can define the type as an update-query L-ADT as follows:
\[
	\textit{Set}_V=(\textit{addSet}(v)_{v\in V}\cup \{\textit{readSet}\}, 2^V\cup\{\bot\}, (2^V,\subseteq,\cup),\emptyset,\tau_{\textit{Set}},\delta_{\textit{Set}}){}.
\]
where $\subseteq$ and $\cup$ are the usual operators on sets, 
$\delta_{\textit{Set}}(z,a)=z$ if $a=\textit{readSet}$ and $\bot$ otherwise, 
and $\tau_{\textit{Set}}(z,a)=z\cup\{v\}$ if $a=\textit{addSet}(v)$ and~$z$ otherwise.

It is easy to see that $(2^V,\subseteq,\cup)$ is a join semi-lattice 
and the L-ADT $\textit{Set}_V$ satisfies the sequential (add-only) \textit{set} specification.
%

\myparagraph{Abort flag.}
%
An \textit{abort-flag} object stores a boolean flag that can only be raised from 
$\bot$ to~$\top$.
Formally, the LADT $\textit{AF}$ is defined as follows:
\[
\textit{AF}=\left(\{\textit{abort},\textit{check}\},\{\bot,\top\},(\{\bot,\top\},\sqsubseteq^\textit{AF},\sqcup^\textit{AF}),\bot,\tau_\textit{AF},\delta_\textit{AF}\right)
\] 
where $\bot\sqsubseteq^\textit{AF}\top$,
$\tau_\textit{AF}(z,\textit{abort})=\delta_\textit{AF}(z,\textit{abort})=\top$, and $\tau_\textit{AF}(z,\textit{check})=\delta_\textit{AF}(z,\textit{check})=z$.

\myparagraph{Conflict detector.}
The \textit{conflict-detector} abstraction~\cite{AE14} exports operation
$\textit{check}(v)$, $v\in V$, that may return $\textit{true}$ (``conflict''), 
or $\textit{false}$ (``no conflict'').
The abstraction respects the following properties:
\begin{itemize}
\item If no two  $\textit{check}$ operations have different inputs,
  then no operation can return $\textit{true}$.

\item  If  two  $\textit{check}$ operations have different inputs,   
then they cannot both return $\textit{false}$.
\end{itemize}
A conflict detector can be specified as an L-ADT defined as follows: 
\[
\textit{CD}=\left(\textit{check}(v)_{v\in V},\{\true,\false\},(V\times \{\top,\bot\},\sqsubseteq^\textit{CD},\sqcup^{\textit{CD}}),\bot,\tau_\textit{CD},\delta_\textit{CD}\right)
\] 
where
\begin{itemize}
\item $\bot\sqsubseteq^{CD} \top$; $\forall v\in V$, $\bot\sqsubseteq^{CD} v$ and
$v\sqsubseteq^{CD} \top$; $\forall v, v'\in V$, $v\neq v' \Rightarrow v
  \not\sqsubseteq^{CD} v'$;
\item $\tau_\textit{CD}(z,\textit{check}(v))=v$ if $z=\bot$ or $z=v$, and $\tau_\textit{CD}(z,\textit{check}(v))=\top$ otherwise; 
\item  $\delta_\textit{CD}(z,\textit{check}(v))=\textit{true}$ if $z=\top$ and $\textit{false}$ otherwise.
\end{itemize}
Also, we can see that $v \sqcup^\textit{CD} v'= v'$ if $v=v'$ or $v=\bot$, and $\top$
otherwise.  

\begin{theorem}
  \label{th:cd}
Any interval-linearizable implementation of \textit{CD} is a conflict detector.
\end{theorem}
\begin{proof}
 Consider any execution of an interval-linearizable implementation of \textit{CD}. 
Let $S$ be the corresponding interval-sequential history.  

For any two  $\textit{check}(v)$ and $\textit{check}(v')$, $v\neq v'$,
in $S$, the response to one of these operations must appear
\emph{after} the invocations of both of them. 
Hence, one of the outputs must be computed on a value greater than the join of the two proposals, equal to $\top$. 
Therefore, if both operations return, at least one of the them must
return $\true$.

The state used to compute the output must be a join of some invoked
operations, hence operations can only return $\textit{true}$ 
if not all $\textit{check}$ operations share the same input.
\end{proof}

\section{Applications}
\label{sec:applications}

Many ADTs do not have commutative
operations and, thus, do not belong to L-ADT.
Moreover, many distributed programming abstractions do not have a
sequential specification at all and, thus, cannot be defined as ADTs,
needless to say as L-ADTs. 

However, as we show,  certain such objects 
can be implemented from L-ADT objects.
As L-ADTs are naturally composable, the resulting implementations can
be seen as using a single (composed) L-ADT object.
By using a reconfigurable version of this L-ADT object, we obtain a
reconfigurable implementation.
In our constructions we omit  talking about reconfigurations
explicitly: to perform an operation on the configuration component of
the system state, a process simply proposes it to the underlying RLA
(see, e.g., Figure~\ref{Alg:ladt}).  

Our examples are atomic snapshots~\cite{AADGMS93}, commit-adopt~\cite{Gaf98} and safe agreement~\cite{BG93b}. 
%
%


\subsubsection*{Atomic Snapshots}

%
An $m$-position atomic-snapshot memory maintains an array of $m$
positions and exports two operations, 
$\textit{update}(i,v)$, where $i\in\{1,\dots,m\}$ is a location in
the array and $v\in V$---the value to be written, that returns a
predefined value $\text{ok}$  
and $\textit{snapshot}()$ that returns an $m$-vector of elements in $V$.
Its sequential specification stipulates that every
$\textit{snapshot}()$ operation returns a vector that contains, in each index $i\in\{1,\dots,m\}$, the value of the 
last preceding $\textit{update}$ operation on the $i^{\textit{th}}$
position (or a predefined initial value, if there is no such
\textit{update}).

\myparagraph{Registers using $\textit{MR}_{\mathbb{N}\times V}$.}
We first consider the special case of a single register ($1$-position atomic
snapshot).
We describe its implementation from a \textit{max-register},
assuming that the set of values $V$ is totally-ordered with relation $\leq^V$.
Let $\leq^{\textit{reg}}$ be a total order on $\mathbb{N}\times V$  (defined
lexicographically, first on $\leq$ and then, in case of equality, on $\leq^V$). 
Let $\textit{MR}$ be a max-register defined on $(\mathbb{N}\times V,\leq^{\textit{reg}})$.

The idea is to associate each written value $\textit{val}$ with a \emph{sequence number} 
$\textit{seq}$ and to store them in $\textit{MR}$ as a tuple $(\textit{seq},\textit{val})$. 
To execute an operation $\textit{update}(v)$, the process first reads~$\textit{MR}$ to
get the ``maximal'' sequence number $s$ written to $\textit{MR}$ so far.
Then it writes~$(s+1,v)$ back to $\textit{MR}$.
Notice that multiple processes may use $s+1$ in their
\textit{update} operations, but only for concurrent operations.
Ties are then broken by choosing the maximal value in the second
component in the tuple. 
A \textit{snapshot} operation simply reads $\textit{MR}$ and returns the value
in the tuple.

Using any reconfigurable linearizable implementation of $\textit{MR}$
(Theorem~\ref{th:crdt}), we obtain a reconfigurable implementation of
an atomic (linearizable) register.
Intuitively, all values returned by \textit{snapshot} (read) operations on
$\textit{MR}$ can be totally ordered based on the corresponding sequence
numbers (ties broken using $\leq^V$), which gives the order of \textit{reads}
in the corresponding sequential history $S$.

\ignore{
By construction, $S$ is legal: every read returns the value of the
last preceding write.
Moreover, as only concurrent \textit{updates} can use the same sequence number
and the \textit{snapshot} operations are ordered respecting the sequence numbers,
$S$ complies with the real-time precedence of the original history. 
We delegate the complete proof to the more general case of an
$m$-position snapshot. 

The idea is to associate the written value with a timestamp to prioritize update operations 
according to their relative execution order. To determine which timestamp to use in an update operation, 
we can simply use an L-ADT counter to ensure that if $\textit{op}_1$ precedes $\textit{op}_2$,
then the timestamp used by $\textit{op}_1$ is striclty smaller than the one used for $\textit{op}_2$. 
For this, it suffices to increment and then read the counter to get an adequate timestamp.
}

\myparagraph{Atomic snapshots.}
Our implementation of an $m$-position atomic snapshot (depicted in Figure~\ref{Alg:MWMR}) 
is a straightforward generalization of the register implementation described above. 
Consider the L-ADT defined as the product of $m$ max-register L-ADTs.   
In particular, the partial order of the L-ADT is the product of $m$ (total) orders 
$\leq^{\textit{snap}}$: $\leq^{\textit{reg}_1}\times\cdots\times
\leq^{\textit{reg}_m}$. 

We also enrich the interface of the type with a new query operation
$\textit{readAll}$ that returns the vector of $m$ values found in the
$m$ max-register components.
Note that the resulting type is still an update-query L-ADT, and thus, 
by Theorem~\ref{th:crdt}, we can use a reconfigurable linearizable
implementation of this type, let us denote it by $\textit{MRset}$. 

To execute $\textit{update}(v,i)$ on the  
implemented atomic snapshot, a process performs a read on the $i^{th}$
component of $\textit{MRset}$ to get sequence number $s$ of the
returned tuple and perfroms $\textit{writeMax}((s+1,v))$ on the $i^{th}$
component.
To execute a snapshot, the process performs $\textit{readAll}$ on
$\textit{MRset}$ and returns the vector of the second element of each 
item of the array.

Similarly to the case of a single register, the results of all \textit{snapshot} operations
can be totally ordered using the $\leq^{\textit{snap}}$ total order on the
returned vectors.        
Placing the matching \textit{update} operation accordingly, we get an
equivalent sequential execution that respects the atomic snapshot specification.

\begin{figure}
\hrule \vspace{1mm}
 {\small
\setcounter{linenumber}{0}
\begin{tabbing}
 bbbb\=bbbb\=bbbb\=bbbb\=bbbb\=bbbb\=bbbb\=bbbb \=  \kill

\textbf{operation} \textit{update}$(i,v)$
\commentline{update register $i$ with $v$}\\
\nnll\label{line:MWMR:readCounter}\> $(s,-) := \textit{MRset}[i].\textit{readMax}$\\
\nnll\label{line:MWMR:writeVal}\> $\textit{MRset}[i].\textit{writeMax}((s+1,v))$\\
\\
\textbf{operation} \textit{snapshot}$()$\\
\nnll\label{line:MWMR:update}\> $r := \textit{MRset}.\textit{readAll}$  \\  
\nnll\label{line:MWMR:sendback}\> \textbf{return} $\textit{snap}$ \textbf{with} $\forall i\in\{1,\dots,m\}, r[i]=(-,\textit{snap}[i])$
\end{tabbing}
}
\vspace{-1.5mm}
 \hrule
\caption{Simulation of an $m$-component atomic snapshot using an L-ADT.}
\label{Alg:MWMR}
\end{figure}

\begin{theorem}
Algorithm in Figure~\ref{Alg:MWMR} implements an $m$-component MWMR atomic snapshot.
\end{theorem}

\begin{proof}
Let us start by providing the linearization order for the atomic snapshot, 
derived from the linearization order provided to the calls to the underlying $\textit{MRset}$ object. 
First, we simply associate \textit{snapshot} operations to the linearization point of their $\mathit{readAll}$
call to $\textit{MRset}$. For \textit{update} operations, two distinct cases are considered.
If the \textit{writeMax} call to $\textit{MRset}[i]$ modifies the state of the max-register object, 
then the \textit{update} is associated with the linearization point of the \textit{writeMax} call.
Otherwise, there must exist another \textit{update} operation already linearized between the 
linearization points of the \textit{readMax} and \textit{writeMax} calls to $\textit{MRset}[i]$. 
Indeed, assume it is not the case, thus, the state of the max-register did not evolve between the 
two calls. But, as the \textit{writeMax} call uses a state strictly greater than the one 
returned by the \textit{readMax} call, it modifies the max-register state --- A contradiction.
Hence, we can linearize the operation to just before an already linearized \textit{update} 
operation on the same index.

The linearization order clearly respects the order of the operations, as it is inherited from a valid linearization order. 
Hence, we only have to show that it also respects the atomic snapshot specification. 
For this, consider any snapshot operation returning the $m$-component array~$\mathit{snap}$ 
and any~$i\in\{1,\dots,m\}$. The preceding \textit{update} operation with index~$i$ 
in the linearization order can only correspond to the last \textit{update} operation 
which modified the max-register. Indeed, all \textit{update} operations which do not 
modify the state are followed by another \textit{update} operation with the same index.
Hence, $\mathit{snap}[i]$ is indeed the value of the preceding \textit{update} operation with index~$i$, 
if any, and $\bot$ otherwise.
\end{proof}

\subsubsection*{The Commit-Adopt Abstraction}
Let us take a more elaborated example, the commit-adopt abstraction~\cite{Gaf98}. 
It is defined through a single operation  $\textit{propose}(v)$, where $v$ belongs to some input domain $V$. 
The operation returns a couple $(\textit{flag},v)$ with $v\in V$ and $\textit{flag}\in\{\textit{commit},\textit{adopt}\}$, 
so that the following conditions are satisfied:
\begin{itemize}
\item \textbf{Validity:} If a process returns $(\_,v)$, then $v$ is the input of some process.
\item \textbf{Convergence:} If all inputs are $v$, then all outputs are $(\textit{commit},v)$.
\item \textbf{Agreement:} If a process returns $(\textit{commit},v)$, then all outputs must be of type $(\_,v)$.
\end{itemize} 



We assume here that $V$, the set of values that can be proposed to the
commit-adopt abstraction, is totally ordered.     
The assumption can be relaxed at the cost of a slightly more
complicated algorithm (by replacing the max register with a set object for example).

Our implementation of (reconfigurable) commit-adopt uses 
a \textit{conflict-detector} object~$\textit{CD}$ (used to detect
distinct proposals), 
a max-register $\textit{MR}_V$ (used to write non-conflicting
proposals), and an \emph{abort flag} object $\textit{AF}$.

Our commit-adopt implementation is presented in Figure~\ref{Alg:CA}. 
In its \textit{propose} operation, a process first accesses the
\textit{conflict-detector} object $\textit{CD}$
(line~\ref{line:CA:CD}).
Intuitively, the conflict detector makes sure that committing processes share a common proposal.

If the object returns $\false$ (no conflict detected), the process writes its
proposal in the max-register $\textit{MR}_V$ 
(line~\ref{line:CA:MaxW}) and then checks the abort flag $\textit{AF}$.
If the check operation  returns~$\bot$, then the proposed value is
returned with the \emph{commit} flag (line~\ref{line:CA:Commit}).
Otherwise, the same value is returned with the \emph{adopt} flag (line~\ref{line:CA:Abort}).

If a conflict is detected ($\textit{CD}$ returns $\true$), 
then the process executes the \textit{abort} operation 
on $\textit{AF}$ (line~\ref{line:CA:RaiseAF}).
Then the process reads the \textit{max-register}.
If a non-$\bot$ value is read (some value has been previously written
to $\textit{MR}$), the process adopts that value (line~\ref{line:CA:AdoptVal}).
Otherwise, the process adopts its own proposed value (line~\ref{line:CA:AdoptProp}). 

\begin{figure}
\hrule \vspace{1mm}
 {\small
\setcounter{linenumber}{0}
\begin{tabbing}
 bbbb\=bbbb\=bbbb\=bbbb\=bbbb\=bbbb\=bbbb\=bbbb \=  \kill

\textbf{operation} \textit{propose}$(v)$\\
\nnll\label{line:CA:CD}\> \textbf{if}
$\textit{CD}.\textit{check}(v)=\textit{false}$ \textbf{then} \commentline{check conflicts}\\
\nnll\label{line:CA:MaxW}\>\> $\textit{MR}_V.\textit{writeMax}(v)$\\
\nnll\label{line:CA:Abort}\>\> \textbf{if}  $\textit{AF}.\textit{check}=\top$ \textbf{then} \textbf{return}  $(\textit{adopt},v)$  \commentline{adopt the input}\\
\nnll\label{line:CA:Commit}\>\> \textbf{else return}  $(\textit{commit},v)$  \commentline{commit proposal}\\
\nnll\label{line:CA:Conflict}\> \textbf{else}  \commentline{Try to abort in case of conflict}\\
\nnll\label{line:CA:RaiseAF}\>\> $\textit{AF}.\textit{abort}$  \commentline{raise abort flag}\\
\nnll\label{line:CA:MaxR}\>\> $\textit{val} := \textit{MR}_V.\text{readMax}$\\
\nnll\label{line:CA:AdoptProp}\>\> \textbf{if}  $\textit{val}=\bot$ \textbf{then}  \textbf{return}  $(\textit{adopt},v)$  \commentline{adopt the input}\\
\nnll\label{line:CA:AdoptVal}\>\> \textbf{else return}  $(\textit{adopt},\textit{val})$  \commentline{adopt the possibly committed value}
\end{tabbing}
}
\vspace{-1.5mm}
 \hrule
\caption{Commit-adopt implementation using L-ADTs.}

\label{Alg:CA}
\end{figure}

\begin{theorem}
Algorithm in Figure~\ref{Alg:CA} implements commit-abort. 
\end{theorem}

\begin{proof}
%
The validity property is trivially satisfied as processes output a couple containing either their proposal, 
if they return in lines~\ref{line:CA:Abort},~\ref{line:CA:Commit} or~\ref{line:CA:AdoptProp}, 
or, if they return in line~\ref{line:CA:AdoptVal}, 
another process proposal previously written to the max-register $\textit{MR}_V$ in line~\ref{line:CA:MaxW}. 

To prove convergence, consider an execution in which all processes share the same input~$v$. 
Hence, all inputs given to the conflict detector $\textit{CD}$ are identical, and so, according to 
conflict detector specification, all processes obtain $\textit{false}$ as output from $\textit{CD}$.
Therefore, as it can only be done in line~\ref{line:CA:RaiseAF}, 
no process calls an $\textit{abort}$ on the abort flag $\textit{AF}$.
This implies that the state of $\textit{AF}$ remains equal to its initial state $\bot$, 
and thus, that $\textit{check}$ calls on $\textit{AF}$ in line~\ref{line:CA:Abort} return~$\bot$. 
Therefore, all processes return in line~\ref{line:CA:Commit}, and so, all outputs are $(\textit{commit},v)$. 

For agreement, consider an execution in which some process $p$ returns with $(\textit{commit},v)$. 
So $p$ returns in line~\ref{line:CA:Commit} and failed the test in line~\ref{line:CA:Abort}. 
Let $\tau$ be the linearization time of the corresponding $\mathit{check}$ on $\textit{AF}$.
Therefore, at time $\tau$, no process may have reached line~\ref{line:CA:MaxR} yet. 
Moreover, $p$ must have written $v$ to $\textit{MR}_V$ before time $\tau$. 
Hence, no process may return in line~\ref{line:CA:AdoptProp}. 
Now, recall that all processes obtaining $\mathit{false}$ from a conflict detector must share 
the same proposal. Hence, all processes returning in lines~\ref{line:CA:Abort} or~\ref{line:CA:Commit}, 
do so with the same proposal as~$p$. It also implies that all values written to 
$\textit{MR}_V$ are equal to $v$, and hence, that all processes returning 
in line~\ref{line:CA:AdoptVal} adopt $v$ for their output.
\end{proof}

\subsubsection*{The Safe Agreement Abstraction}
Another popular shared-memory abstraction is  \emph{safe agreement}~\cite{BG93b}. 
It is defined through a single operation $\textit{propose}(v)$, 
$v\in V$ (we assume that $V$ is totally ordered). 
The operation returns a value $v\in V$ or a special value $\bot\not\in
V$ so that the following conditions are satisfied:
\begin{itemize}
\item \textbf{Validity:} Every non-$\bot$ output has been previously proposed.
\item \textbf{Agreement:} All non-$\bot$ outputs are identical.
\item \textbf{Non-triviality:} If all participating processes return, then at least one returns a non-$\bot$ value.
\end{itemize} 

Our implementation of safe agreement (Figure~\ref{Alg:SA}) uses two (add-only) sets denoted
$\mathit{In}$ and $\mathit{Out}$ (Section~\ref{sec:ladt}) 
and a max-register $\textit{MR}_V$.
 
The \textit{propose} operation consists of two phases. 
In the first phase (lines~\ref{line:SA:EnterDW}--\ref{line:SA:ExitDW})
that we call the \emph{doorway protocol}, processed first add their identifier to $\mathit{In}$. 
Then processes read $\mathit{MR}_V$, and, if $\bot$ is returned, they then write their proposal to the max-register.
Finally, they exit the doorway by adding their identifier to the $\mathit{Out}$ set. 

In the second phase (lines~\ref{line:SA:CheckDW}--\ref{line:SA:occupiedDW}), 
processes first read the $\mathit{In}$ and $\mathit{Out}$ sets.
If the two sets match, then processes read and return the content of the max-register.
Otherwise, the special value $\bot$ is returned. 

Intuitively, the doorway protocol is used to ensure that processes can check 
if a process might be poised to write to the max-register.
The second phase consists of returning the max-register content if no process may still write to it.

\begin{figure}
\hrule \vspace{1mm}
 {\small
\setcounter{linenumber}{0}
\begin{tabbing}
 bbbb\=bbbb\=bbbb\=bbbb\=bbbb\=bbbb\=bbbb\=bbbb \=  \kill

\textbf{operation} \textit{propose}$(v)$\\
\nnll\label{line:SA:EnterDW}\> $\mathit{In}.\textit{addSet}(\textit{id})$  \commentline{enter the doorway}\\
\nnll\label{line:SA:Prop}\>  \textbf{if} $\textit{MR}_V.\textit{readMax}=\bot$ \textbf{then} $\textit{MR}_V.\textit{writeMax}(v)$ \commentline{write proposal if empty}\\
\nnll\label{line:SA:ExitDW}\>  $\mathit{Out}.\textit{addSet}(\textit{id})$  \commentline{exit the doorway}\\

\nnll\label{line:SA:CheckDW}\> \textbf{if} $\mathit{In}.\textit{readSet}=\mathit{Out}.\textit{readSet}$ \textbf{then} \commentline{check doorway}\\
\nnll\label{line:SA:RetMaxRegister}\>\> \textbf{return}  $\textit{MR}_V.\textit{readMax}$\\
\nnll\label{line:SA:occupiedDW}\> \textbf{else}\\
\nnll\label{line:SA:abort}\>\> \textbf{return}  $\bot$  
\end{tabbing}
}
\vspace{-1.5mm}
 \hrule
\caption{Safe agreement implementation using L-ADTs for process with identifier $\textit{id}$.}

\label{Alg:SA}
\end{figure}

\begin{theorem}
Algorithm in Figure~\ref{Alg:SA} implements safe agreement. 
\end{theorem}

\begin{proof}
The validity property is trivially satisfied as any non-$\bot$
returned value must be an input value written to the max-register in line~\ref{line:SA:Prop}.

Let us now show that the agreement property is also verified. 
Let $\tau$ be the time at which the max-register $\textit{MR}_V$ is first written to.
The set of processes writing to $\textit{MR}_V$, let us call it~$S$, thus corresponds to 
the set of processes that already checked their test in line~\ref{line:SA:Prop} at time~$\tau$.
At time $\tau$ the set object $\mathit{In}$ already contains $S$ and no process reached 
line~\ref{line:SA:CheckDW} yet to read the $\mathit{In}$ set. 
Therefore, all reads to $\mathit{In}$ contains $S$. 
Now let us assume that some process~$p$ returns with a non-$\bot$ output, 
hence, in line~\ref{line:SA:RetMaxRegister}. 
Therefore, $p$ must have successfully passed the test in line~\ref{line:SA:CheckDW} 
and so with a read of the set object $\mathit{Out}$ containing $S$. 
Hence, after all processes that may write to $\textit{MR}_V$ already exited the doorway. 
Therefore, $p$ returns the final state reached by $\textit{MR}_V$ during the execution. 
Thus, non-$\bot$ outputs are all equal to the final state of $\textit{MR}_V$.

For the non-triviality property, consider an execution in which all participating processes return 
and let $p$ be the last process to write its identifier to the $\textit{Out}$ set (line~\ref{line:SA:ExitDW}). 
At that time, the $\textit{In}$ and $\textit{Out}$ sets are both equal to the set of participating processes. 
Thus, $p$ returns the final state of $\mathit{MR}_V$ (line~\ref{line:SA:RetMaxRegister}).
Now, let us assume by contradiction that $\bot$ is the final state of~$\mathit{MR}_V$. 
Hence, $p$ must have successfully passed the test in line~\ref{line:SA:Prop} and wrote its value to~$\mathit{MR}_V$. 
Therefore, the final state of $\mathit{MR}_V$ is greater than the proposal of $p$ --- A contradiction.
\end{proof}

\section{Related Work}
\label{sec:related}

%
\myparagraph{Lattice agreement.}
Attiya et al.~\cite{lattice-hagit} introduced the (one-shot) lattice agreement abstraction and, in the shared-memory context, described a wait-free reduction of lattice agreement to atomic snapshot.
\ignore{
has been introduced by . Unlike consensus, which is impossible even if one process is faulty, lattice agreement is a decidable decision problem in asynchronous settings, and several wait-free algorithms for shared memory and message-passing systems have been proposed in the literature. Attiya et al. \cite{lattice-hagit} shows the equivalence between lattice agreement and atomic snapshot, proposing for the first time a transformation of a decision problem weaker than consensus, giving in upper bounds on shared objects. As for complexity, the best known solution to the problem has time complexity $\mathcal{O}(n\log{}n)$, where $n$
is  the number of processes taking part to the  computation. 
In message-passing systems, lattice agreement can be solved in asynchronous systems when a majority of processes is correct.
}
Falerio et al.~\cite{gla} introduced the long-lived version of lattice
agreement (adopted in this paper) and described an asynchronous
message-passing implementation of lattice agreement assuming a majority
of correct processes, with $\mathcal{O}(n)$ time complexity (in terms
of message delays) in a system of $n$ processes. 
Our RLA implementation in Section~\ref{sec:implementation} builds upon
this algorithm.  

\myparagraph{CRDT.}
Conflict-free replicated data types (CRDT) were introduced by Shapiro et
al.~\cite{crdt} for eventually synchronous replicated services.
The types are defined using the language of join semi-lattices and
assume that type operations are partitioned in updates and queries.
Falerio et al.~\cite{gla} describe a ``universal'' construction of a
linearizable CRDT from lattice agreement.
Skrzypczak et al.~\cite{crdt-practice} argue that avoiding consensus
in such constructions may bring performance gains. 
In this paper, we consider a more general class of  types (L-ADT) that are ``state-commutative''
but not necessarily ``update-query'' and leverage the recently
introduced criterion of interval-linearizability~\cite{interval-linearizability} for
\emph{reconfigurable} implementations of
L-ADTs using RLA.   

\ignore{
Generalized lattice agreement is then used to build a replicated state machine for data structures supporting commutative operations as conflict-free replicated data types (CRDTs \cite{crdt}). In CRDTs, all operations commute such that they can be concurrently executed without coordination. Notably\cite{gla} shows that generalized lattice agreement on top of CRDTs (initially introduced for eventual consistency) allows any two replicated states to be comparable, guaranteeing linearizability \cite{HW90}.
}


\myparagraph{Reconfiguration.}
\ignore{
The problem of object reconfigurability has been widely explored in
recent years. Many works in this area focused on emulations of
read/write registers. Those works can be classified depending on the
type of reconfiguration assumed.
}
\textit{Passive reconfiguration}~\cite{BaldoniBKR09,AttiyaCEKW19}
assumes that replicas enter and leave the system under an explicit
\textit{churn model}:
if the churn assumptions are violated, consistency is not guaranteed.
In the \textit{active reconfiguration} model, processes explicitly propose
configuration updates, e.g., sets of new process members. 
Early proposals, such as RAMBO~\cite{rambo} focused on read-write
storage services and used consensus to ensure that the clients agree
on the evolution of configurations. 

Recent solutions~\cite{dynastore,parsimonious,smartmerge,freestore,SKM17-reconf} 
propose an asynchronous reconfiguration by replacing consensus with weaker abstractions 
capturing the minimal coordination required to safely modify the system configuration. 
Moreover, Freestore~\cite{freestore} proposes a modular solution to derive 
interchangeable consensus-based and asynchronous reconfiguration.

\myparagraph{Asynchronous reconfiguration.}
Dynastore~\cite{dynastore} was the first solution emulating a
reconfigurable atomic read/write register without consensus:
clients can asynchronously propose  incremental additions or removals
to the system configuration.
Since proposals commute, concurrent proposals are collected together
without the need of deciding on a total order.
Assuming $n$ proposals, a Dynastore client
might, in the worst case, go through $2^{n-1}$ candidate configurations before converging to
a final one.
%
\ignore{
Dynastore assures atomicity in all runs and liveness only if the
number of configurations is finite and the set of crashed processes
and processes whose removal is pending is a minority in the current
configuration and in any concurrent proposal\footnote{Note that these
liveness conditions are necessary in the case of active
reconfiguration \cite{SpiegelmanKM17}.}.
}
Assuming a run with a total number of configurations $m$, complexity is
$\mathcal{O}(min(mn, 2^{n}))$. 

SmartMerge~\cite{smartmerge} allows for reconfiguring not only the
system membership but also its quorum system, excluding possible
undesirable configurations.
SmartMerge brings an interesting idea of using 
an external reconfiguration service based on lattice
agreement~\cite{gla},
which allows us to reduce the number of traversed
 configurations to $O(n)$.
However, this solution assumes that this ``reconfiguration lattice''  is  always available and
 non-reconfigurable (as we showed in this paper, lattice agreement is a powerful
 tool that can itself be used to implement a large variety of
 objects).

Gafni and Malkhi \cite{parsimonious} proposed the \textit{parsimonious
speculative snapshot} task based on the commit-adopt
abstraction~\cite{Gaf98}.
Reconfiguration,  built on top of the proposed abstraction, has
complexity  $\mathcal{O}(n^2)$: $n$ for the traversal and $n$ for the
complexity of the parsimonious speculative snapshot implementation.
%
%
Spiegelman, Keidar and Malkhi \cite{SKM17-reconf} improved this work
by proposing an optimal solution with time complexity $\mathcal{O}(n)$ by 
obtaining an amortized (per process) time complexity $\mathcal{O}(1)$ for speculative snapshots operations.

\section{Concluding Remarks}
\label{sec:conc}

To conclude, let us briefly discuss the complexity of our solution to the reconfiguration problem 
and give an overview of how our solution could be further extended. 

\myparagraph{Round-trip complexity.}
The main complexity metric considered in the literature 
is the maximal number of communication round-trips needed to complete a reconfiguration 
when $c$ operations are concurrently proposed. In the worst case, each time a 
round of requests is completed in our algorithm, a new state is affecting
$T_p$ or $\mathit{obj}_p$,  and hence we have at most $c$
round-trips.
Note that a round might be interrupted by receiving a greater 
committed state at most $c$ times as committed states are 
totally ordered joins of proposed states.
We are aware of only one other optimal solution with linear 
round-trip complexity, proposed Spiegelman et al.~\cite{SKM17-reconf}.
In their solution, the maximal number of round-trips is at least $4c$,
twice more than ours. 
This has to do with the use of a shared memory simulation preventing to read and write 
at the same time and preventing from sending requests to distinct configurations in parallel.
Moreover, they also use a similar interruption mechanism.

Querying multiple configurations at the same time might increase the round-trip delay 
as we need to wait for more responses. Still, we believe that when the number of requests 
scales with a constant factor, this impact is negligible.

\myparagraph{Message complexity.}
%
As in earlier solutions, messages are of linear size in the number of  distinct proposed configurations 
or collect operations on the implemented object.

The number of exchanged messages depends on 
the configuration lattice.
With at most~$k$ members per configuration, 
each client may send at most $k*2^n$ messages per round as there are,
in the worst case,  exponentially many configurations to query.
But this upper bound may be reached only if 
joins of proposed configurations do not share replicas.
%
%
We expect, however, that in most cases the concurrently proposed
configurations have large overlaps: configuration updates are
typically gradual. 
For example, when a configuration is defined as a set of updates (added and
removed replicas),  clients may send 
at most $k+\Delta*n$ requests per round, where $\Delta$
the number of replicas added per proposal. 
For small $\Delta$, 
the total number of messages is of order $k$.

%
An interesting question is whether we can construct a composite
complexity metric that combines the number of messages a process sends
and the time it takes to complete a \textit{propose} operation.
Indeed, one may try to find dependencies between accessing few
configurations sequentially versus accessing many configurations in parallel.

\myparagraph{Complexity trade-offs.}
If the cost of querying many configurations in parallel outweigh the cost of contacting fewer configurations sequentially, 
we can use the approach from~\cite{SKM17-reconf}. 
Intuitively, it boils down to solving an instance of generalized lattice agreement on
the configurations and then querying the produced configurations,
there can be $O(c)$ of them, where $c$ is the number of concurrently
proposed configurations.

\ignore{
A lighter modification to the RLA protocol may consist in assuming timing 
constraints and to wait sufficiently long until most responses are
received,
while waiting for responses from quorums only when no new information
is received and an operation may return.
Such a modification may yield efficiency gains in practice, as clients should be less constrained 
by slow responses while increasing the number of distinct inputs expected to discover 
per round.
\textbf{PK: did not get this one}
}

Objects with ``well-structured''
its lattices can be implemented very efficiently.
Take, for example, the totally ordered lattice of a max-register.
In this case, processes can directly return the state stored in $\mathit{LearnLB}$ in 
line~\ref{line:rgla2:setAdoptLB}.
Indeed, not returning a committed state might 
only violate the \emph{consistency} property. But if states are
already totally ordered, 
then the \emph{consistency} property always holds.
Therefore, in the absence of reconfiguration calls, operations can return in a single round trip. 
It is in general interesting to investigate how the lattice structure 
might be leveraged.

\def\noopsort#1{} \def\No{\kern-.25em\lower.2ex\hbox{\char'27}}
  \def\no#1{\relax} \def\http#1{{\\{\small\tt
  http://www-litp.ibp.fr:80/{$\sim$}#1}}}




\end{document}